\documentclass[11pt,a4paper]{article}

\usepackage{amsmath}
\usepackage{amsthm}
\usepackage{amssymb}
\usepackage{amsfonts}
\usepackage{graphicx}
\usepackage{authblk}

\usepackage{latexsym}
\usepackage{graphicx}
\usepackage{marvosym}
\usepackage{hyperref}
\usepackage{float}
\usepackage{color}
\def\lie{{\mathcal{L}}}

\def\hi{\mathcal{H}}
\def\o{\mathcal{O}}

\def\pv{\partial_{v}}
\def\vh{\mathcal{V}_{\hh}}

\newcommand{\li}{\mbox{$\lie \mkern-9.5mu /$}}

\def\f12{\frac 1 2}

\def\s{\mathcal{S}}

\def\a{\alpha}
\def\b{\beta}
\def\ga{\gamma}

\def\hh{\mathcal{H}}
\def\m{\mathcal{M}}

\newcommand{\nabb}{\mbox{$\nabla \mkern-13mu /$\,}}
\newcommand{\lapp}{\mbox{$\triangle \mkern-13mu /$\,}}
\newcommand{\epsi}{\mbox{$\epsilon \mkern-7.4mu /$\,}}
\newcommand{\gi}{\mbox{$g \mkern-8.8mu /$\,}}
\newcommand{\di}{\mbox{$d \mkern-9.2mu /$\,}}
\newcommand{\divv}{\mbox{$\div \mkern-16mu /$\,\,}}

\def\f12{\frac 1 2}

\def\div{\text{div}}

\newtheorem{definition}{Definition}[section]

\newtheorem{remark}{Remark}[section]
\newtheorem{lemma}{Lemma}[section]
\newtheorem{theorem}{Theorem}[section]
\newtheorem{proposition}{Proposition}[section]

\newtheorem{corollary}{Corollary}[section]

\begin{document}
\title{On a foliation-covariant elliptic operator on null hypersurfaces}


\author[1,2]{Stefanos Aretakis}
\affil[1]{Princeton University, Department of Mathematics, Fine Hall, NJ 08544, USA}
\affil[2]{Institute for Advanced Study, Einstein Drive, Princeton, NJ 08540, USA}
\date{October 4, 2013}

\maketitle

\begin{abstract}

We introduce a new elliptic operator on null hypersurfaces of four-dimensional Lorentzian manifolds. This operator depends on the first and second fundamental forms of the sections of a foliation of the null hypersurface and its novelty originates from its \textit{covariant} transformation under change of foliation. It thus provides at any point an elliptic structure intimately connected with the geometry of the null hypersurface, independent of the choice of a  specific section through that point. No analytic or algebraic symmetries or other conditions are imposed on the metric. The spectral properties of this elliptic operator are relevant to  the evolution of the wave equation, and in particular, the existence of conservation laws along null hypersurfaces. 

\end{abstract}

\tableofcontents

\section{Introduction}
\label{sec:Introduction}


Elliptic operators on null hypersurfaces have played a major role in a wide spectrum of problems in Lorentzian geometry (see, for example, \cite{christab,feffermanhirachi,galloway2000}). Refoliating a null hypersurface has also been proved to be a very strong tool in several contexts (see, for example, \cite{arichbondimass,klr,sauter2008}).

In this paper, we introduce a new elliptic operator on null hypersurfaces of a general four-dimensional Lorentzian manifold which exhibits a specific covariance property under change of foliation of the null hypersurfaces. As is shown in the companion paper \cite{aretakisglue}, the spectral properties of this operator play an important role in the evolution of the wave equation on Lorentzian manifolds, and in particular, in the existence of conservation laws along null hypersurfaces.

Let $\hh$ be a null hypersurface of a four-dimensional Lorentzian manifold $(\m,g)$. We assume that $\hh$ can be foliated by sections $\big(S_{v}\big)_{v\in\mathbb{R}}$ diffeomorphic to $\mathbb{S}^{2}$, i.e.~embedded two-dimensional Riemannian submanifolds diffeomophic to $\mathbb{S}^{2}$ and intersecting transversally the null generators of $\hh$. As we shall see in Section \ref{sec:TheGeometryOfNullHypersurfaces},  any foliation $\s$ of $\hh$ is completely determined by the choice of a section $S_{0}$ of $\hh$,  a null geodesic vector field $L_{geod}$ on $\hh$ and a function $\Omega\in C^{\infty}(\hh)$. We then write 
\[\s=\Big\langle S_{0}, L_{geod}, \Omega  \Big\rangle. \]
Given a foliation $\s=\big(S_{v}\big)_{v\in\mathbb{R}}=\left\langle S_{0}, L_{geod}, \Omega\right\rangle$ of a null hypersurface $\hh$, we define the following  second order linear operator $\mathcal{A}^{\s}:C^{\infty}(\hh)\rightarrow\mathbb{R}$ 
\begin{equation}
\begin{split}
\mathcal{A}^{\mathcal{S}}\psi=&\lapp\psi+\Big[ Z-2\nabb\log\phi\Big]\cdot\nabb\psi+\left[\phi\cdot\lapp\frac{1}{\phi}-\nabb\log\phi\cdot Z+\Omega^{-2}\cdot w\right]\cdot\psi. 
\end{split}
\label{aintro}
\end{equation}
\begin{equation}
\begin{split}
Z=&2\zeta^{\sharp}+\nabb\log\Omega^{2},\\w=2\divv\big(\Omega^{2}\zeta\big)+&\Omega^{2}\cdot L_{geod}\, (\Omega tr\underline{\chi})+\frac{1}{2}(\Omega tr\chi) (\Omega tr\underline{\chi}),
\end{split}
\label{zwintro}
\end{equation}
where $\lapp$ and $\nabb$ denote the induced Laplacian and gradient on the sections $S_{v}$, respectively, and where $\zeta$ is the torsion, $tr\chi,\, tr\underline{\chi}$ are the null mean curvatures and $\phi$ is the conformal factor of $S_{v}$. See Section \ref{sec:TheGeometryOfNullHypersurfaces} for the relevant definitions.

Clearly, the operator $\mathcal{A}^{\s}$ is tangential to the sections $S_{v}$ of $\s$ and the restriction  \[\mathcal{A}_{v}^{\s}:=\left.\mathcal{A}^{S}\right|_{S_{v}}:C^{\infty}(S_{v})\rightarrow\mathbb{R}\]  is an elliptic operator on the Riemannian manifold $S_{v}$ which depends only on the geometry of the foliation $\s$, i.e~the first and the second fundamental forms of the sections of $\s$ with respect to the ambient manifold $\m$.  
 
We define the following space
\begin{equation}
\vh=\Big\{f\in C^{\infty}(\hh)\, :\, f \text{ is constant along the null generators of }\hh\Big\}.
\label{eq:vh}
\end{equation}
Let now $\mathcal{S}'=\big(S'_{v'}\big)_{v'\in\mathbb{R}}=\big\langle S_{0}', L_{geod}', \Omega'  \big\rangle$ be anotherfoliation of $\hh$. Since, $L_{geod},L_{geod}'$ satisfy the geodesic equation, there exists a function $f\in \vh$ such that \[L_{geod}'=f^{2}\cdot L_{geod}.\] One can consider the associated geometric elliptic operator $\mathcal{A}^{\s'}$ defined in an identical way to $\mathcal{A}^{\s}$ by simply replacing the geometric quantities associated to the foliation $\s$ with those of the foliation $\s'$. 

Given a point $p\in\hh$ there exist two sections $S_{v}\in \s$ and $S'_{v'}\in \s'$ such that $p\in S_{v}\cap S'_{v'}$.  Since the operators $\left.\mathcal{A}^{\s}\right|_{p}$ and $\left.\mathcal{A}^{\s'}\right|_{p}$ are tangential to $S_{v}$ and $S'_{v'}$, respectively, the only way to compare $\left.\mathcal{A}^{\s}\right|_{p}$ and $\left.\mathcal{A}^{\s'}\right|_{p}$ is by identifying the sections $S_{v}$ and $S'_{v'}$ via the flow of the null generators, or equivalently, by restricting to functions $\Psi\in\vh$.
 \begin{figure}[H]
   \centering
		\includegraphics[scale=0.109]{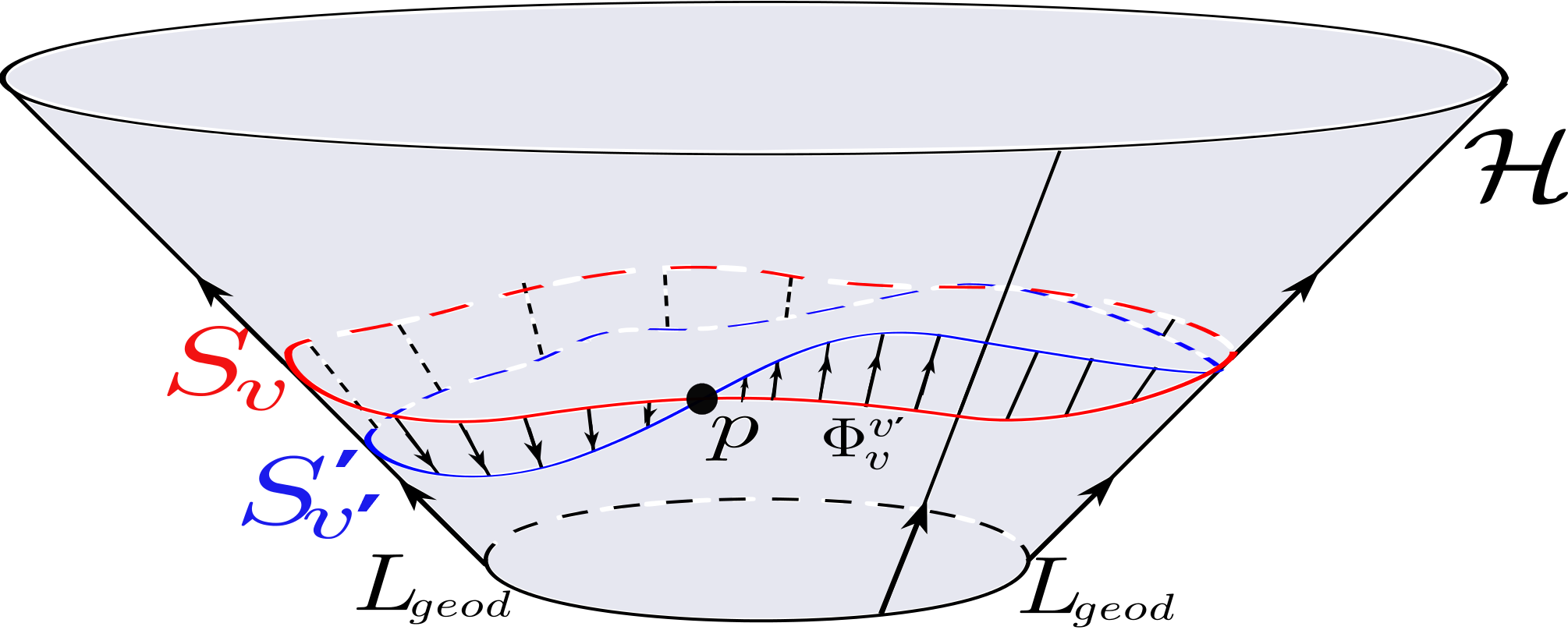}
	\label{fig:nullsection11212jlkj12121etromono}
\end{figure}
The main result of the present paper is the following covariance transformation of the operator $\mathcal{A}^{\s}$ under change of foliation:  For all functions $\Psi\in\vh$ we have
\begin{equation}
\mathcal{A}^{\mathcal{S}}\Psi=\frac{1}{f^{2}}\cdot\mathcal{A}^{\mathcal{S}'}
\left(f^{2}\cdot\Psi\right)
\label{eq:intromainequation}
\end{equation}
on $\hh$. 
Therefore, if we restrict to foliations for which the geodesic vector field $L_{geod}$ is fixed (and hence $f=1$), and hence allow only the ``initial'' section $S_{0}$ and the null lapse function $\Omega$ to change then we obtain that at each $p\in S_{v}\cap S'_{v'}$  the restrictions $\mathcal{A}_{v}^{\s}$ and $\mathcal{A}_{v'}^{\s'}$ are exactly the same  modulo identifying the sections $S_{v}$ and $S'_{v'}$ via $\Phi_{v}^{v'}$, i.e.
\[\big(\Phi_{v}^{v'}\big)^{*}\mathcal{A}_{v'}^{\s'}=\mathcal{A}_{v}^{\s} \ \text{  at  }p, \]
where the pullback operator $\big(\Phi_{v}^{v'}\big)^{*}\mathcal{A}_{v'}^{\s'}$ is the operator on $S_{v}$  defined such that $\Big(\big(\Phi_{v}^{v'}\big)^{*}\mathcal{A}_{v'}^{\s'}\Big)(\psi)=\mathcal{A}_{v'}^{\s'}\Big(\big(\Phi_{v}^{v'}\big)_{*}\psi\Big),$
for all $\psi\in C^{\infty}\big(S_{v}\big)$. Here the diffeomorphism $\Phi_{v}^{v'}:S_{v}\rightarrow S'_{v'}$ is defined such that if $q\in S_{v}$, then $\Phi_{v}^{v'}(q)$ is the intersection of $S_{v'}'$ and the null generator passing through $q$. 

For simplicity, we restrict to null hypersurfaces with spherical sections. Note that our method makes the use of the topology explicit and hence other topologies can be readily treated. Furthermore, our result can be generalized to all higher dimensions.  No conditions, apart  from smoothness, are imposed on the background metric. The precise result is given by Theorem \ref{maintheorem} of Section \ref{sec:ChangeOfFoliationII}.

A brief outline of the paper is as follows: In Section \ref{sec:TheGeometryOfNullHypersurfaces} we introduce the basic geometric set-up and in Section \ref{sec:ChangeOfFoliationII} we introduce the elliptic operator $\mathcal{A}^{\mathcal{S}}$ and prove the main result of the paper. In Section \ref{sec:KillingHorizons} we present some applications in the context of  black hole backgrounds. Finally, in Section \ref{sec:EpilogueEllipticOperatorsOnNullHypersurfaces} we present a general discussion of foliation-covariant elliptic operators on null hypersurfaces which might be relevant to other contexts.

\section{The geometry of null hypersurfaces}
\label{sec:TheGeometryOfNullHypersurfaces}

Let $\hh$ be a null hypersurface in a four-dimensional Lorentzian manifold $(\m,g)$. Then $\hh$ is generated (ruled) by null geodesics, the so-called \textit{null generators}, whose tangent $L$ is normal to $\hh$. A \textit{section} $S$ of $\hh$ is a two-dimensional submanifold of $\hh$ which intersects each null generator of $\hh$ transversally and is thus manifestly a Riemannian manifold equipped with the induced Riemannian metric which we will denote by $\gi$. We assume that all sections of $\hh$ are diffeomorphic to the 2-sphere $\mathbb{S}^{2}$. In fact, we will later construct an explicit diffeomorphism from $S$ to $\mathbb{S}^{2}$ which will be very important for our applications. 
 \begin{figure}[H]
   \centering
		\includegraphics[scale=0.114]{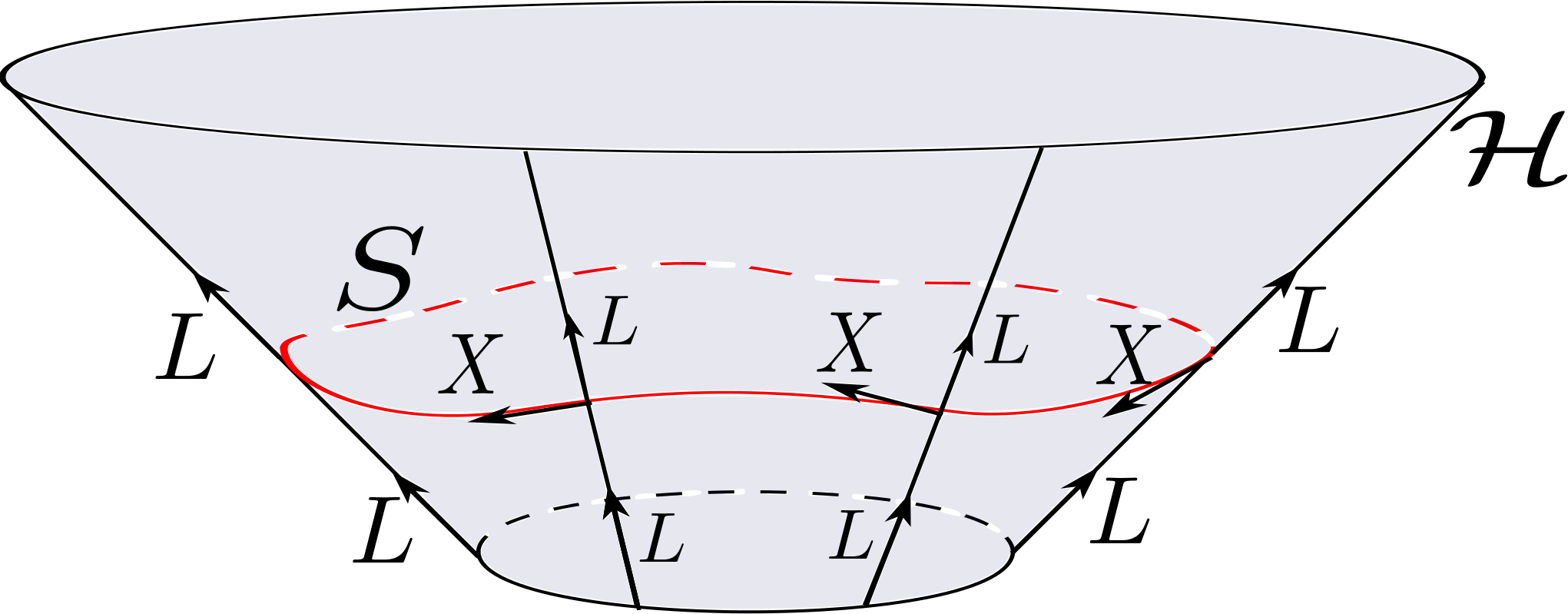}
	\label{fig:nullsection}
\end{figure}

\newpage

\noindent\textbf{Null Foliations}
\bigskip
 
A \textit{foliation} $\mathcal{S}$ of $\hh$ is a collection of sections $S_{v}$, smoothly varying in $v$, such that $\cup_{v}S_{v}=\hh$. We will show that any foliation is uniquely determined by the choice of one section, say $S_{0}$, the choice of a null tangential to $\hh$ vector field $\left.L_{geod}\right|_{S_{0}}$ restricted on $S_{0}$ and a function $\Omega$ on $\hh$. Indeed, we extend $\left.L_{geod}\right|_{S_{0}}$  to a  null vector field tangential to the null generators of $\hh$ such that
\[\nabla_{{L}_{geod}}{{L}_{geod}}=0.\]We then define the vector field
\begin{equation}
L=\Omega^{2}\cdot L_{geod}
\label{eq:l}
\end{equation}
on $\hh$ and consider the affine parameter $v$ of $L$ such that 
\[Lv=1, \text{ with } v=0 \text{ on } S_{0}.\]
The level sets $S_{v}$  of $v$ on $\hh$ are precisely the leaves of the foliaton $\mathcal{S}$. We use the notation
\begin{equation}
\mathcal{S}=\left\langle S_{0}, \left.L_{geod}\right|_{S_{0}}, \Omega \right\rangle.
\label{foliationdef}
\end{equation}
We also define $\underline{L}_{geod}$ on $\hh$ to be null normal to $S_{v}$, conjugate to $\hh$ and normalized such that
\[g\big(L_{geod},\underline{L}_{geod}\big)=-\Omega^{-2}.\]

\bigskip

\noindent\textbf{The Diffeomorphisms $\Phi_{v}$}

\bigskip

We can construct a  diffeomorphism $\Phi_{v}$ from any sphere $S_{v}$ to  $S_{0}$ as follows: If $p\in S_{v}$, then  $\Phi_{v}(p)\in S_{0}$ is defined to be the intersection of $S_{0}$ and the null generator of $\hh$ passing through $p$. We can also consider a diffeomorphism $\Phi$ from $S_{0}$ to $\mathbb{S}^{2}$ and compose $\Phi_{v}$ with $\Phi$ to obtain a diffeomorphism from $S_{v}$ to $\mathbb{S}^{2}$. 

Let $(\theta^{1},\theta^{2})$ be coordinates on $S_{0}$. The diffeomorphisms $\Phi_{v}$ allow us to construct a coordinate system $(v,\theta^{1},\theta^{2})$ as follows: The point $p\in\hh$ is assigned the coordinates $(v,\theta^{1},\theta^{2})$ if $p\in S_{v}$ and the coordinates of $\Phi_{v}(p)$ are precisely $(\theta^{1},\theta^{2})$.
 \begin{figure}[H]
   \centering
		\includegraphics[scale=0.06]{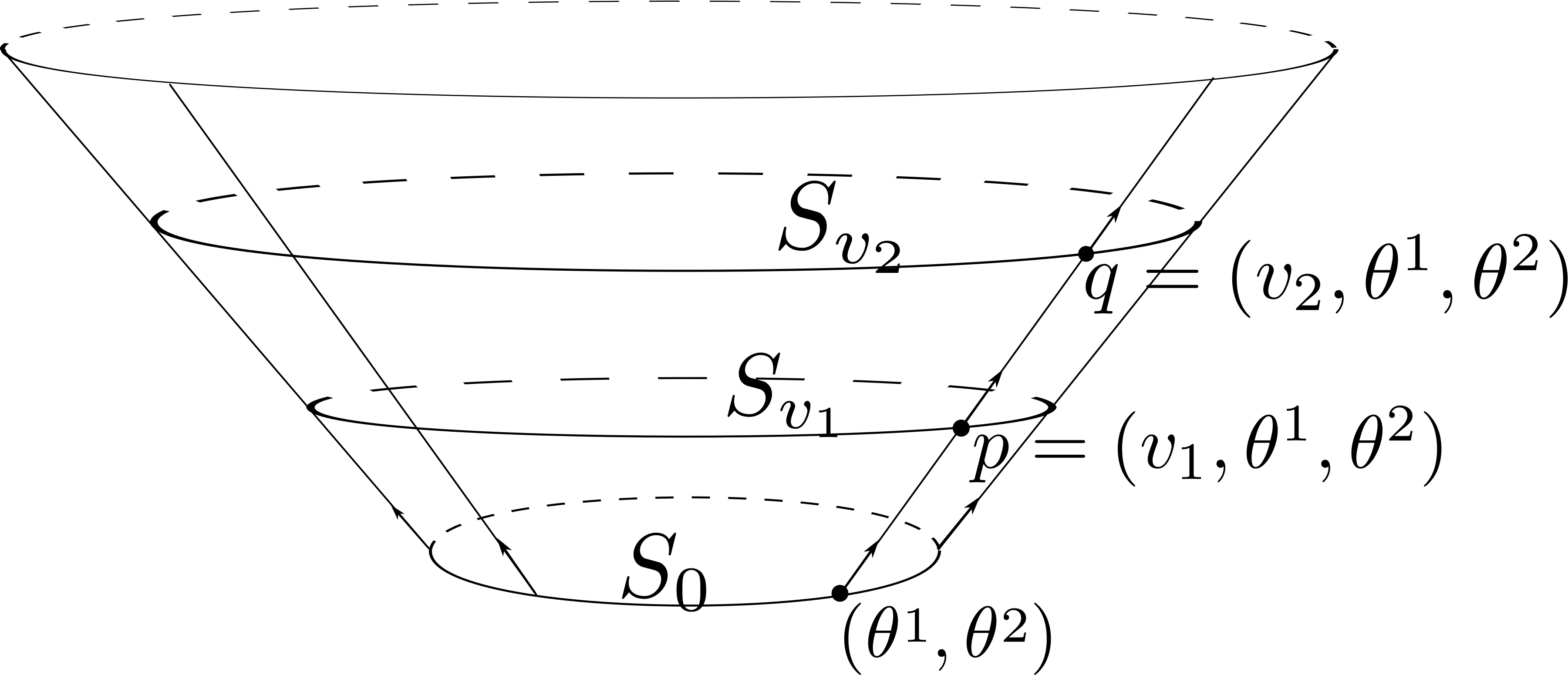}
	\label{fig:p45h}
\end{figure}Moreover, the diffeomorpshisms $\Phi_{v}$ allow us to equip all surfaces $S_{v}$ with the standard round metric which we denote by $\gi_{\mathbb{S}^{2}}$. 


\bigskip

\noindent\textbf{Null Frames}
\bigskip
 
\noindent If $\left\{e_{1},e_{2}\right\}=\big(e_{A}\big)_{A=1,2}$ is an arbitrary frame on the spheres $S_{v}$, then we have the following null frames:
\begin{itemize}
	\item 
\textbf{Geodesic frame:} $(e_{1},e_{2},L_{geod}, \underline{L}_{geod})$,
\item
\textbf{Equivariant frame:} $(e_{1},e_{2},L,\underline{L})$,
\item
\textbf{Normalized frame:} $(e_{1},e_{2}, e_{3},e_{4}).$

\end{itemize}
where \[e_{3}=\Omega \underline{L}_{geod}=\frac{1}{\Omega}\underline{L}, \ \ \ \ e_{4}=\Omega L_{geod}=\frac{1}{\Omega}{L}.\]
\begin{figure}[H]
   \centering
		\includegraphics[scale=0.135]{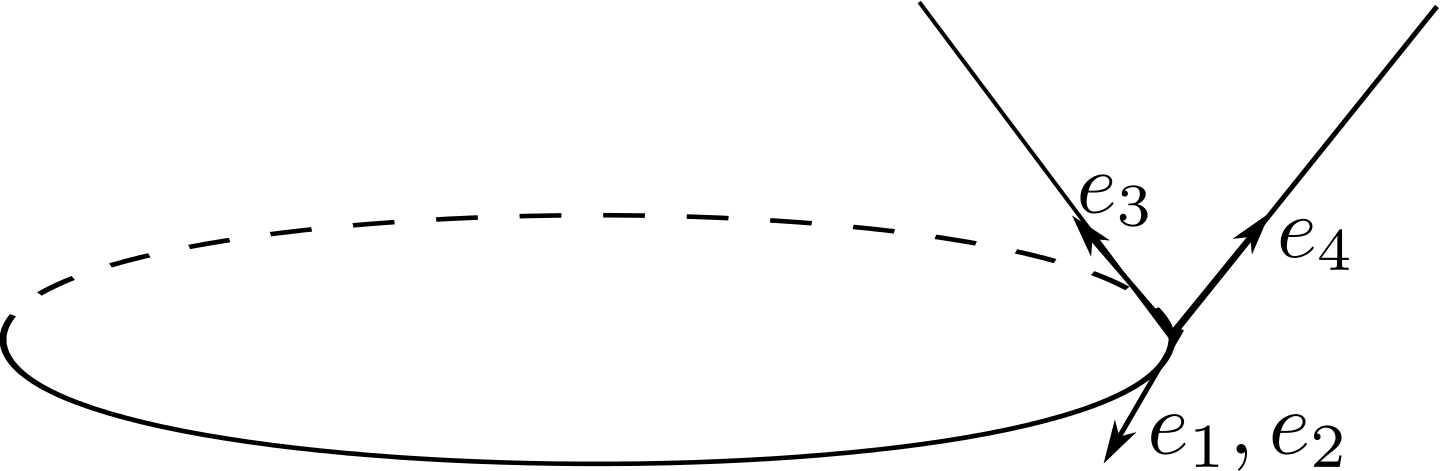}
	\label{fig:p3344}
\end{figure} 
Note that $e_{3},e_{4}$ satisfy the normalization properties
\begin{equation}
g(e_{3},e_{4})=-1
\label{eq:}
\end{equation}
and \begin{equation}
g\big(L_{geod},\underline{L}\big)=-1.
\label{normalizationtransversal}
\end{equation}
Note also that  $L=\pv$, where $\pv$ denotes the vector field with respect to the coordinate system $(v,\theta^{1},\theta^{2})$.

\bigskip

\noindent\textbf{The conformal geometry and conformal factor}
\bigskip

The conformal class of $\gi$ contains a unique representative (metric) $\hat{\gi}$ such that $\sqrt{\hat{\gi}}=\sqrt{\gi_{\mathbb{S}^{2}}}$, where $\gi_{\mathbb{S}^{2}}$ is as defined above. Equivalently, $\gi$ and $\hat{\gi}$ are such that the induced volume forms on $S_{v}$ are equal. Since $\gi$ and $\hat{\gi}$ are conformal there is a conformal factor such that $\gi=\phi^{2}\cdot \hat{\gi}$. Then, $\sqrt{\gi}=\phi^{2}\sqrt{\hat{\gi}}=\phi^{2}\sqrt{\gi_{\mathbb{S}^{2}}}$ and hence $\phi^{2}= \frac{\sqrt{\gi}}{\sqrt{\gi_{\mathbb{S}^{2}}}}. $
Therefore,
\begin{equation}
\phi=\frac{\sqrt[4]{\gi}}{\sqrt[4]{\gi_{\mathbb{S}^{2}}}}.
\label{phi}
\end{equation}
Note that $\phi$ is a smooth function on the sphere $S_{v}$ and does not depend on the choice of the coordinate system. Note also that for a spherically symmetric metric we have $\phi=r$, where $r$ is the radius.

\bigskip

\noindent\textbf{Connection Coefficients}

\bigskip

We consider the \textbf{normalized frame} $(e_{1},e_{2},e_{3},e_{4})$ defined above.  We define the connection coefficients with respect to this frame to be the smooth functions $\Gamma^{\lambda}_{\mu\nu}$
such that 
\[\nabla_{e_{\mu}}e_{\nu}=\Gamma^{\lambda}_{\mu\nu}e_{\lambda}, \ \ \lambda, \mu, \nu\in\left\{1,2,3,4\right\}\]
Here $\nabla$ denotes the connection of the spacetime metric $g$. We are mainly interested in the case where at least one of the indices $\lambda,\mu,\nu$ is either 3 or 4 (otherwise, we obtain the Christoffel symbols with respect to the induced metric $\gi$).  Following \cite{DC09,christab}, these coefficients  are completely determined by the following components:
\medskip

\noindent\textbf{The components} $\chi,\underline{\chi},\eta,\underline{\eta}, \omega,\underline{\omega},\zeta$:
\smallskip
\begin{equation}
\begin{split}
&\chi_{AB}=g(\nabla_{A}e_{4},e_{B}), \ \ \ \ \ \underline{\chi}_{AB}=g(\nabla_{A}e_{3},e_{B}),\\
&\ \eta_{A}=g(\nabla_{3}e_{4},e_{A}),\ \ \ \ \  \ \,  \underline{\eta}_{A}=g(\nabla_{4}e_{3},e_{A}),\\
&\ \ \ \omega=-g(\nabla_{4}e_{4},e_{3}), \ \ \ \  \underline{\omega}=-g(\nabla_{3}e_{3},e_{4}),\\
& \ \ \ \ \ \ \ \ \ \ \ \ \ \ \ \ \ \ \ \ \  \zeta_{A}=g(\nabla_{A}e_{4},e_{3})
\end{split}
\label{concoef}
\end{equation}
where $\big(e_{A}\big)_{A=1,2}$ is an arbitrary frame on the spheres $S_{v}$ and $\nabla_{\mu}=\nabla_{e_{\mu}}$. Note that $\underline{\zeta}=-\zeta$. The covariant tensor fields $\chi,\underline{\chi},\eta,\underline{\eta}$ and $\zeta$ are only defined on $T_{x}S_{v}$. \textbf{We can naturally extend these to tensor fields to be defined on $T_{x}\m$ by simply letting their value to be zero if they act on $e_{3}$ or $e_{4}$. Such tensor fields will in general be called \underline{$S$-tensor fields}. Note that a vector field is an $S$-vector field if it is tangent to the spheres $S_{v}$.}

The connection coefficients $\Gamma$ can be recovered by the following relations:
\begin{equation}
\begin{split}
\nabla_{A}e_{B}=\nabb_{A}e_{B}+&\chi_{AB}e_{3}+\underline{\chi}_{AB}e_{4},\\
\nabla_{3}e_{A}=\nabb_{3}e_{A}+\eta_{A}e_{3},& \ \ \ \nabla_{4}e_{A}=\nabb_{4}e_{A}+\underline{\eta}_{A}e_{4},\\
\nabla_{A}e_{3}=\underline{\chi}_{A}^{\ \ \sharp B}e_{B}+\zeta_{A}e_{3},& \ \ \ \nabla_{A}e_{4}={\chi}_{A}^{\ \ \sharp B}e_{B}-\zeta_{A}e_{4},\\
\nabla_{3}e_{4}=\eta^{\sharp A}e_{A}-\underline{\omega}e_{4},& \ \ \ \nabla_{4}e_{3}=\underline{\eta}^{\sharp A}e_{A}-\omega e_{3},\\
\nabla_{3}e_{3}=\underline{\omega}e_{3},& \ \ \ \nabla_{4}e_{4}=\omega e_{4},
\end{split}
\label{connectioncoef}
\end{equation}

\bigskip

\noindent\textbf{Curvature Components}

\bigskip

We next decompose the Riemann curvature $R$ in terms of the normalized null frame. First, we define the following components, which contain at most two S-tangential components (and hence at least 2 null components):
\begin{equation}
\begin{split}
\a_{AB}=R_{A4B4},&\ \ \ \ \underline{\a}_{AB}=R_{A3B3},\\
\beta_{A}=R_{A434},&\ \ \ \ \underline{\beta}_{A}=R_{A334},\\
\rho= R_{3434}, & \ \ \ \sigma=\frac{1}{2}\epsi^{AB}R_{AB34}.
\end{split}
\label{curvcompdeflist}
\end{equation}
Note that $R(\cdot,\cdot, e_{3},e_{4})$, when restricted on $T_{x}S_{v}$, is an antisymmetric form and hence collinear to the volume form $\epsi$ on $S_{v}$. Furthermore, if 
\[(*R)_{3434}=*\rho,\]
then $*\rho= 2\sigma$. Here,  the dual $*R$ of the Riemann curvature is defined to be the (0,4) tensor:
\[(*R)_{\a\b\gamma\delta}=\epsilon_{\mu\nu\a\beta}\,R^{\mu\nu}_{\ \ \  \gamma\delta}.\]
Clearly, the (0,2) $S$-tensor fields $\a,\underline{\a}$ are symmetric. Note that if the Einstein equations $Ric(g)=0$ are satisfied, then all the remaining curvature components can be expressed in terms of the above components. 

\newpage

\noindent\textbf{Remarks:}
\medskip

\textbf{1.} Recall that the second fundamental form of a manifold $S$ embedded in a manifold $\m$ is defined to be the symmetric $(0,2)$ tensor field $II$ such that for each $x\in S$ we have
\[II_{x}:T_{x}S\times T_{x}S\rightarrow (T_{x}S)^{\perp}, \]
where the $\perp$ is defined via the decomposition $T_{x}\m=T_{x}S\oplus(T_{x}S)^{\perp}$. Specifically, if $X,Y\in T_{x}S$ then \[II_{x}(X,Y)=\big(\nabla_{X}Y\big)^{\perp}\] and hence\[\nabla_{X}Y=\nabb_{X}Y+II(X,Y).\]
Here $\nabb$ denotes the induced connection on $S$ (which is taken by projecting the spacetime connection $\nabla$ on $T_{x}S$).

 The $S$-tensor fields $\chi,\underline{\chi}$ give us the projections of $II_{AB}$ on $e_{3}$ and $e_{4}$, respectively. Indeed
 \[II(X,Y)=\chi(X,Y)e_{3}+\underline{\chi}(X,Y)e_{4}.\]
 
 For this reason we will refer to 
 $\chi,\underline{\chi}$ as the \textit{null second fundamental forms} of $S_{v}$.   One can easily verify that $\chi$ and $\underline{\chi}$ are symmetric (0,2) $S$-tensor fields. Indeed, a simple calculation shows that if $X,Y$ are $S$-tangent vector fields then
 \begin{equation*}
 \begin{split}
  \chi(X,Y)-\chi(Y,X)=g(e_{4},[X,Y]) ,\\ \underline{\chi}(X,Y)-\underline{\chi}(Y,X)=g(e_{3},[X,Y]).
 \end{split}
 \end{equation*}
 Hence, $\chi,\underline{\chi}$ are symmetric if and only if $[X,Y]\perp e_{3}$ and $[X,Y]\perp e_{4}$ and thus if and only if $\left\langle e_{3},e_{4}\right\rangle^{\perp}\ni [X,Y]\in TS_{v}$. The symmetry of $\chi,\underline{\chi}$ is thus equivalent to the integrability of the orthogonal complement $\left\langle e_{3},e_{4}\right\rangle^{\perp}$.
 
  Furthermore, we can decompose $\chi$ and $\underline{\chi}$ into their trace and traceless parts by
 \begin{equation}
\chi=\hat{\chi}+\frac{1}{2}(tr\chi)\gi, \ \ \ \
\underline{\chi}=\hat{\underline{\chi}}+\frac{1}{2}(tr\underline{\chi})\gi.
\label{tracechi}
\end{equation}
The trace of the $S$-tensor fields $\chi,\underline{\chi}$ (and more general $S$-tensor fields) is taken with respect to the induced metric $\gi$. The trace $tr\chi$ is known as the  \textit{expansion} and the component $\hat{\chi}$ is called the \textit{shear} of $S_{v}$ with respect to $\hh$.

\medskip

\textbf{2.} Note also that $\omega= \nabla_{4} (\log\Omega)=$ and $\underline{\omega}= \nabla_{3} (\log\Omega)$.

\medskip

\textbf{3.} Let $X$ be a vector tangential to a given sphere $S_{v}$ at a point $x$. Then, if we extend $X$ along the null generator $\gamma$ of $C_{u}$ passing through $x$ according to the Jacobi equation $[L,X]=[\Omega e_{4},X]=0$, then we obtain an S-tangent vector field along $\gamma$. Note that in this case we obtain
\[\nabla_{4}X=\nabla_{X}e_{4}+\big(\nabla_{X}\log\Omega\big)e_{4}.\] On the other hand, if we simply extend $X$ such that $[e_{4},X]=0$ then, although $X$ will be tangential to $C_{u}$, $X$ will not be tangential to the sections $S_{v}$ of $C_{u}$. This is because the sections $S_{v}$ are the level sets of the optical functions $u,v$ which in  turn are the affine parameters of the vector fields $\underline{L},L$, respectively. 

\medskip

\textbf{4.} The $S$ 1-form $\zeta$ is known as the \textit{torsion}. 
If $\di$ denotes the exterior derivative on $S_{v}$ then the $S$ 1-forms $\eta,\underline{\eta}$ are related to $\zeta$ via
\begin{equation*}
\begin{split}\eta=\zeta+\di(\log\Omega),\ \ \ \ \underline{\eta}=-\zeta+\di{\log\Omega}.
\end{split}
\end{equation*}
\begin{proof}
Let $X$ be $S$-tangential and extend along the null generator of $C_{u}$ according to the Jacobi equation, then 
\begin{equation}\underline{\eta}(x)=g(\nabla_{4}e_{3},X)=-g(e_{3},\nabla_{4}X)=-g(e_{3},\nabla_{X}e_{4})+\nabla_{X}(\log\Omega)=-\zeta(X)+\nabla_{X}\log\Omega.\label{mmm}\end{equation}
We  similarly show the analogous relation for $\eta$.\end{proof}
The above also imply that
\[\zeta= \frac{1}{2}(\eta-\underline{\eta}), \ \ \ \ \di\log\Omega=\frac{1}{2}(\eta+\underline{\eta}).\]
The 1-forms $\eta,\underline{\eta}$ can be regarded as the torsion of the null hypersurfaces with respect to the geodesic vector fields. Indeed, the previous relations imply
\[\eta_{A}=\Omega^{2}g(\nabla_{A}L_{geod}, \underline{L}_{geod}).\]

\medskip

\textbf{5.} We have
\begin{equation}
\begin{split}
[L,\underline{L}]&=\nabla_{L}\underline{L}-\nabla_{\underline{L}}L=\Omega\cdot\Big(\nabla_{4}(\Omega e_{3})-\nabla_{3}(\Omega e_{4})\Big)\\
&=\Omega^{2}\cdot\Big(\nabla_{4}e_{3}-\nabla_{3}e_{4}+(\nabla_{4}\log\Omega) e_{3}-(\nabla_{3}\log\Omega)e_{4}\Big)\\
&=\Omega^{2}\cdot\Big((\underline{\eta}^{\sharp A}-{\eta}^{\sharp A})e_{A}-\omega e_{3}+\underline{\omega}e_{4}+(\nabla_{4}\log\Omega) e_{3}-(\nabla_{3}\log\Omega)e_{4}\Big)\\
&=\Omega^{2}\cdot \Big(\underline{\eta}^{\sharp}-\eta^{\sharp}\Big)\\
&=-2\Omega^{2}\zeta^{\sharp}.
\end{split}
\label{torsionll}
\end{equation}Hence the torsion $\zeta$ is the obstruction to the integrability of the timelike planes $\left\langle e_{3},e_{4}\right\rangle$ orthogonal to the spheres $S_{v}$.

\medskip

\textbf{6. } Let $\li_{L}$ denote the projection of the Lie derivative $\lie_{L}$ onto the spheres $S_{v}$. The \textbf{first variation formula} then reads
\begin{equation}
\li_{L}\gi=2\Omega\chi,\ \ \ \ \ \li_{L}(\gi^{-1})=-2\Omega\chi^{\sharp\sharp}.
\label{1stvarform}
\end{equation}Note that we use the induced metric $\gi$ to raise and lower indices.
 Hence, since $[L,\partial_{\theta^{i}}]=0$ on $\hh$,
\[ L\big(\gi_{ij}\big)=2\Omega\chi_{ij} \]
and hence
\begin{equation}
L\sqrt{\gi}=\Omega tr\chi\sqrt{\gi}
\label{derdet}
\end{equation}
on $\hh$. Therefore, if $\phi$ denotes the conformal factor of the sections then 
\begin{equation}
L\phi=\frac{1}{2}\Omega\cdot tr\chi\cdot \phi,
\label{lphi}
\end{equation}
since by construction we have that $L\sqrt{\gi_{\mathbb{S}^{2}}}=0$.  Furthermore, the first variational formula immediately implies
\begin{equation}
\hat{\chi}=\frac{1}{2}\phi^{2}\cdot \li_{L}\hat{\gi},
\label{conformalequation}
\end{equation}
 where $\hat{\gi}$ is the representative of the conformal geometry of $\gi$ as defined above. Hence, $\hat{\chi}$ controls the rate of change of the conformal geometry of the sections of $\hh$.

\medskip

\textbf{7. } We denote by $\li_{L},\nabb_{L}$ the projection of $\lie_{L},\nabla_{L}$ on the sections $S_{v}$ and by $\lapp,\nabb$ the induced Laplacian and gradient of $(S_{v},\gi)$, respectively.

\section{Foliation-covariance of the operator $\mathcal{A}$}
\label{sec:ChangeOfFoliationII}

Given a foliation $\mathcal{S}=(S_{v})_{v\in\mathbb{R}}=\left\langle S_{0},L_{geod}, \Omega\right\rangle$ of $\hh$, we define the operator  
\begin{equation}
\begin{split}
\mathcal{A}^{\mathcal{S}}\psi=&\lapp\Psi+\Big[ Z-2\nabb\log\phi\Big]\cdot\nabb\psi+\left[\phi\cdot\lapp\frac{1}{\phi}-\nabb\log\phi\cdot Z+\Omega^{-2}\cdot w\right]\cdot\psi,
\end{split}
\label{a}
\end{equation}
\begin{equation}
\begin{split}
Z=&2\zeta^{\sharp}+\nabb\log\Omega^{2},\\w=2\divv(\Omega^{2}\zeta)+&\Omega^{2}\cdot L_{geod}\, (\Omega tr\underline{\chi})+\frac{1}{2}(\Omega tr\chi) (\Omega tr\underline{\chi}),
\end{split}
\label{zw}
\end{equation}
on $\hh$, where $\zeta$ is the torsion, $tr\chi,\, tr\underline{\chi}$ are the null mean curvatures and $\phi$ is the conformal factor of $S_{v}$ (see \eqref{concoef} for the relevant definitions). 
Clearly, the restriction $\mathcal{A}_{v}^{\mathcal{S}}$ of $\mathcal{A}^{\s}$ on $S_{v}$ is an elliptic operator on $S_{v}$ and its definition depends only on the geometry of the foliation $\s$.  Note also that the first two terms of $w$ are at the level of the curvature since they involve derivatives of the connection coefficients. The main result of this paper is that the  operator $\mathcal{A}^{\mathcal{S}}$ is in fact covariant under change of foliation. Specifically, we show the following

\begin{theorem}Let $\hh$ be a regular null hypersurface of a four-dimensional Lorentzian manifold $(\m,g)$ and let $\vh$ be the space given by \eqref{eq:vh}. 
Let $\big(L_{geod}\big)_{1},\big(L_{geod}\big)_{2}$ be two geodesic vector fields on $\hh$ such that  
\begin{equation}
\big(L_{geod}\big)_{2}=f^{2}\cdot \big(L_{geod}\big)_{1}
\label{f}
\end{equation} for some function $f\in \vh$.  Consider two foliations $\mathcal{S}_{1},\mathcal{S}_{2}$ such that
\begin{equation}
\mathcal{S}_{1}=(S^{2}_{v})_{v\in\mathbb{R}}=\Big\langle S^{1}_{0}, \big(L_{geod}\big)_{1}, \Omega_{1}\Big\rangle
\label{s11}
\end{equation}
and
\begin{equation}
\mathcal{S}_{2}=(S^{1}_{v})_{v\in\mathbb{R}}=\Big\langle S_{0}^{2},\big(L_{geod}\big)_{2}, \Omega_{2}\Big\rangle, 
\label{s21}
\end{equation}
as defined in Section \ref{sec:TheGeometryOfNullHypersurfaces}, 
and let $\mathcal{A}^{\mathcal{S}_{1}},\mathcal{A}^{\mathcal{S}_{2}}$ be the associated elliptic operators given by \eqref{a}. Then, for all functions $\Psi\in\vh$ we have
\begin{equation}
\mathcal{A}^{\mathcal{S}_{1}}\Psi=\frac{1}{f^{2}}\cdot\mathcal{A}^{\mathcal{S}_{2}}
\left(f^{2}\cdot\Psi\right)
\label{maineqcor},
\end{equation}
on $\hh$. 
\label{maintheorem}
\end{theorem}

\begin{proof}
We first consider the case where $f=1$, i.e.~$\big(L_{geod}\big)_{2}=\big(L_{geod}\big)_{1}=L_{geod}$. In this case, in order to simplify the notation, we use the primed and unprimed notation for the two foliations. Specifically, we consider two foliations $\mathcal{S},\mathcal{S}'$ such that
\begin{equation}
\mathcal{S}=(S_{v})_{v\in\mathbb{R}}=\left\langle S_{0}, \left.L_{geod}\right|_{S_{0}}, \Omega_{\mathcal{S}}=1\right\rangle
\label{s1}
\end{equation}
and
\begin{equation}
\mathcal{S}'=(S'_{v'})_{v'\in\mathbb{R}}=\left\langle S_{0}', \left.L_{geod}\right|_{S_{0}'}, \Omega_{\mathcal{S}'}=\Omega\right\rangle
\label{s2}
\end{equation}
and let $\mathcal{A}^{\s},\mathcal{A}^{\s'}$ be the associated elliptic operators given by \eqref{a}, \eqref{o}, respectively. All the quantities with respect to $\mathcal{S}'$, except from $\Omega$, will be primed and all quantites with respect to $\mathcal{S}$ will be unprimed.

Next, given a foliation $\mathcal{S}=\left\langle S_{0},L_{geod},\Omega\right\rangle$, we introduce the auxiliary operator
\begin{equation}
\o^{\mathcal{S}}\psi=\Omega^{2}\cdot\lapp\psi+\Omega^{2}\cdot Z \cdot\nabb\Psi+w\cdot\psi,
\label{o}
\end{equation}
where $Z,w$ are given by \eqref{zw} and $\psi$ is a smooth function on $\hh$. It is a straightforward calculation to confirm that 
\begin{equation}
\mathcal{A}^{\mathcal{S}}\psi=\Omega^{-2}\cdot{\phi}\cdot\o^{\mathcal{S}}\left(\frac{1}{\phi}\cdot \psi\right),
\label{oa}
\end{equation}
on $\hh$, for all $\psi\in C^{\infty}(\hh)$. Since $f=1$, we need to show that for all $\Psi\in\vh$ we have
\begin{equation}
\mathcal{A}^{\s'}\Psi=\mathcal{A}^{\s}\Psi 
\label{fora}
\end{equation}
on $\hh$. Equivalently, in view of \eqref{oa} and since  $\phi'=\phi$ pointwise, it suffices to show that for $\Psi\in\vh$ we have
\begin{equation}
\o^{\s'}\left(\frac{1}{\phi'}\cdot \Psi\right)=\Omega^{2}\cdot\o^{\s}\left(\frac{1}{\phi}\cdot \Psi\right),
\label{toshow}
\end{equation}on $\hh$.  Note that the  operators $\o^{\s},\o^{\s'}$ are given explicitly by:
\[\o^{\s}\psi=\lapp\psi+2\zeta^{\sharp}\cdot\nabb\psi+\left[2\divv\,\zeta^{\sharp}+\partial_{v}( tr\underline{\chi})+\frac{1}{2}( tr\underline{\chi} )( tr\chi)\right]\cdot\psi,\]
and
\[\o^{\s'}\psi=\Omega^{2}\cdot\lapp'\psi+\left[\nabb'\Omega^{2}+2\Omega^{2}\cdot(\zeta')^{\sharp}\right]\cdot\nabb'\psi+\left[2\divv'\,\Big(\Omega^{2}\cdot(\zeta')^{\sharp}\Big)+\partial_{v'}(\Omega tr'\underline{\chi}')+\frac{1}{2}(\Omega tr'\underline{\chi}' )(\Omega tr'\chi')\right]\cdot\psi,\]
for $\phi\in C^{\infty}(\hh)$. By assumption we have $L_{geod}'=L_{geod}$ and so
\[e_{4}=L=\partial_{v}=L_{geod}\]
and
\[  e_{4}'=\Omega e_{4} , \ \  L'=\partial_{v'}=\Omega^{2}L, \]  on $\hh$.
Hence, we immediately find that 
\begin{equation}
\chi'=\Omega\cdot \chi 
\label{chichi}
\end{equation}
on $\hh$. 

Let now $p\in S_{v}\cap S'_{v'}$ and $X\in T_{p}S_{v}$. Clearly $dim\Big(\left\langle X,L\right\rangle\cap T_{p}S'_{v'}\Big)=1$ and hence   there is a point where the line $X+\left\langle L\right\rangle\subset T_{p}\hh$ intersects $T_{p}S'_{v'}$, that is for all $X\in T_{p}S_{v}$ there is a unique $\gamma(X)$ such that 
\begin{equation}
X'=X+\gamma(X)\cdot L\in T_{p}S'_{v'}.
\label{gamm2}
\end{equation}  
Clearly, the correspondence 
\[\Pi:(T_{p}{S_{v}},\gi)\ni X\mapsto X' \in (T_{p}S'_{v'},\gi')  \]
is a linear isomorphism (in fact, an isometry) and so $\gamma$ is an $S$ 1-form. Note that $\Pi(X)$ is simply the projection of $X$ onto $T_{p}S'_{v'}$ and $\gamma(X)$ is the projection of $X'$ onto $\left\langle L\right\rangle$ relative to $S_{v}$.

If $\big(\partial_{\theta^{i}}\big)\in TS_{v},\big(\partial_{\theta^{i}}\big)'\in TS'_{v'},i=1,2,$ denote the coordinate vector fields with respect to $(v,\theta^{1},\theta^{2})$ and $(v',\theta^{1},\theta^{2})$, respectively, then one easily obtains that 
\begin{equation}
\big(\partial_{\theta^{i}}\big)'=\Pi\big(\partial_{\theta^{i}}\big). 
\label{forcoordinatecovariance}
\end{equation}

 We next compute $\gamma$.  Recall that 
\[v=0  \text{ at }S_{0}, \ \ \ v'=0 \text{ at }S'_{0}. \]
Furthermore, 
\[ \partial_{v'}v'=\Omega^{2}\pv v'\Rightarrow \pv v'=\Omega^{-2} \ \text{ on }\hh \]
and if
\[v'(v=0,\theta)=i(\theta)\] then
\[  v'(v,\theta)=f(v,\theta) \, +i(\theta)\ \text{ on }\hh,  \]
where \[f(v,\theta)= \int_{0}^{v}\Omega^{-2}(\overline{v},\theta)\, d\overline{v}.\]
Using $v'$ as a test function in \eqref{gamm2} we immediately obtain
\[\gamma(X)=-\Omega^{2}\cdot Xv'=-\Omega^{2}\cdot X\cdot\nabb (f+i)\text{ on }\hh\]
and so
\begin{equation}
\gamma=-\Omega^{2}\cdot \di\left(f+i\right)\text{ on }\hh.
\label{gamma2}
\end{equation}
We next compute the projection of $\underline{L}'$ 
on $T_{p}S_{v}$ (recall that $\underline{L}'$ is normal to $S'_{v'}$). Suppose  that
\[\underline{L}'=c_{L}\cdot L+c_{\underline{L}}\cdot \underline{L}+P \ \text{ on }\hh,\]
where $P\in T_{p}S_{v}$ is the projection of $\underline{L}'$  on $T_{p}S_{v}$.  
 \begin{figure}[H]
   \centering
		\includegraphics[scale=0.127]{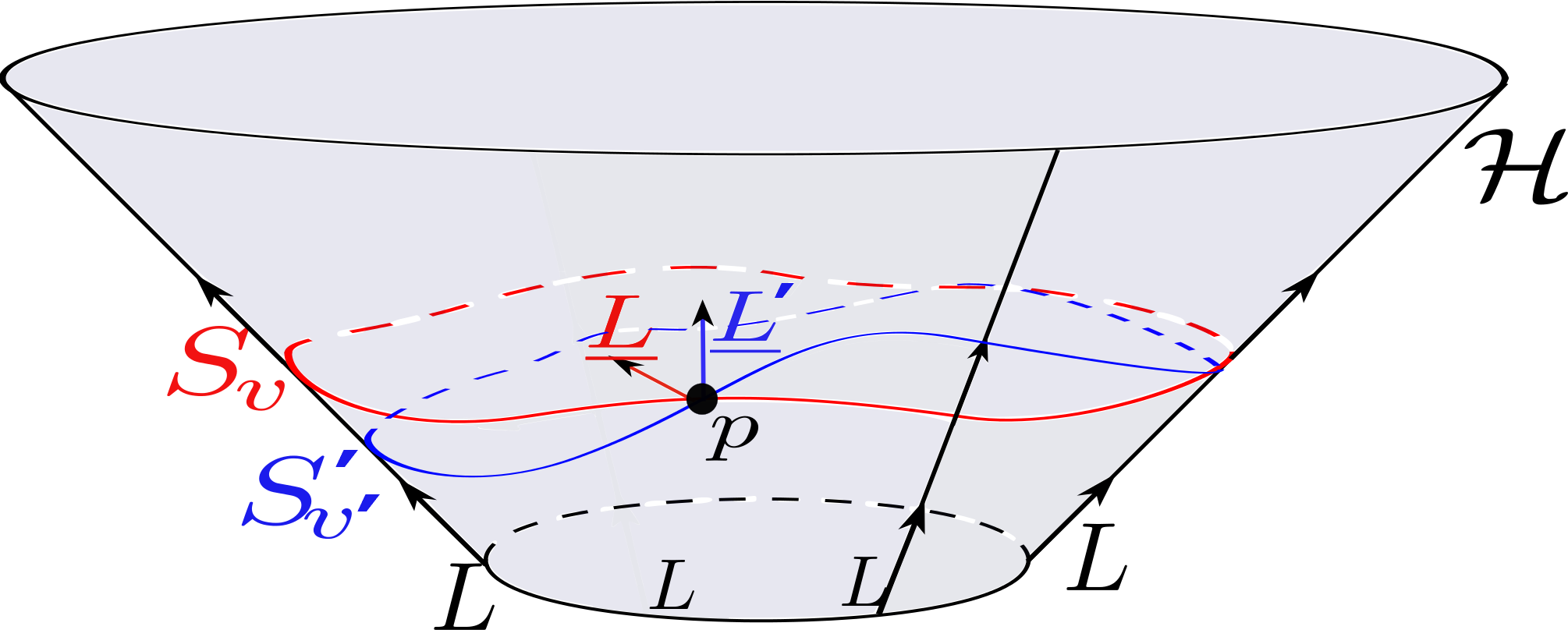}
	\label{fig:nullsection1}
\end{figure}

Since $g(\underline{L}',L)=-1$ we obtain \[c_{\underline{L}}=1 \text{ on }\hh.\] 
Using that $g(\underline{L}',X')=0$ for all $X'\in T_{p}S'_{v'}$ and the relation $X=X'-\gamma(X)\cdot L$ we obtain
\[g(P,X)=\gamma(X) \text{ on }\hh,   \]
which implies that   
\[P_{\flat}=\gamma\]or, equivalently, 
\[P=\gamma^{\sharp}=-\Omega^{2}\cdot\nabb\left(\int_{0}^{v}\Omega^{-2}\right)\text{ on }\hh.  \]
Also, since $g(\underline{L}',\underline{L}')=0$ we obtain 
\[c_{L}=\frac{1}{2}g(P,P)=\frac{1}{2}(\gamma,\gamma) \text{ on }\hh. \]
Therefore,
\begin{equation}
\underline{L}'=\frac{1}{2}(\ga,\ga)\cdot L+\underline{L}+\ga^{\sharp}.
\label{lbar}
\end{equation}
The torsion $\zeta'$ satisfies
\begin{equation*}
\begin{split}
\zeta'(X')=&g(\nabla_{X'}(e_{4})',(e_{3})')=g\Big(\nabla_{X'}\Omega^{-1}\cdot L', \Omega^{-1}\cdot \underline{L}'\Big)\\=&-\Omega\cdot \nabla_{X'}\frac{1}{\Omega}+\frac{1}{\Omega^{2}}g\Big(\nabla_{X'}L',\underline{L}'\Big)
\\=&-\Omega\cdot \nabla_{X'}\frac{1}{\Omega}+\frac{1}{\Omega^{2}}\cdot g\Big(\nabla_{X'}(\Omega^{2}L),\frac{1}{2}\big(\gamma,\gamma\big)\cdot L+\underline{L}+\ga^{\sharp}\Big)
\\=&-\Omega\cdot \nabla_{X'}\frac{1}{\Omega}-\frac{1}{\Omega^{2}}\nabla_{X'}\Omega^{2}+g\Big(\nabla_{X'}L,\frac{1}{2}\big(\gamma,\gamma\big)\cdot L+\underline{L}+\ga^{\sharp}\Big)
\\=&-\Omega\cdot \nabla_{X'}\frac{1}{\Omega}-\frac{1}{\Omega^{2}}\nabla_{X'}\Omega^{2}+g\Big(\nabla_{X'}L,\underline{L}\Big)+g\Big(\nabla_{X'}L,\ga^{\sharp}\Big)\\
=&-(\di'\log\Omega)(X')+\zeta(X')+(\chi\cdot P)(X'),
\end{split}
\end{equation*}
on $\hh$. Note also that $\zeta(X')=\zeta(X)$, $(\chi\cdot \ga^{\sharp})(X')=(\chi\cdot \ga^{\sharp})(X)$ and $\nabb'_{X'}\log\Omega = \nabb_{\gamma(X)\cdot L}\log\Omega+\nabb_X{\log\Omega}$. Therefore, 
\begin{equation}
\zeta'(X')=-\gamma(X)\cdot L\log\Omega-\nabb_{X}\log\Omega+\zeta(X)+(\chi\cdot \gamma^{\sharp})(X).
\label{zetazetatonos}
\end{equation}
We next compute the incoming null second fundamental form
\begin{equation*}\begin{split}
\underline{\chi}'(X',Y')=&g(\nabla_{X'}(e_{3})',Y' )=g(\nabla_{X'}\Omega^{-1}\underline{L'}, Y')=\Omega^{-1}\cdot g(\nabla_{X'}\underline{L}', Y')\\
=& \Omega^{-1}\cdot g\left(\nabla_{X'}\left(\frac{1}{2}\big(\gamma,\gamma\big)\cdot L+\underline{L}+P\right),Y'\right)\\
=&\Omega^{-1}\cdot \left[\frac{1}{2}\big(\gamma,\gamma\big)\cdot g(\nabla_{X'}L,Y')+g(\nabla_{X'}\underline{L},Y')+g(\nabla_{X'}P,Y') \right]. 
\end{split}\end{equation*}
Note that $g(\nabla_{X'}L,Y')=g(\nabla_{X}L,Y)=\chi(X,Y)$ and also the following
\begin{equation*}\begin{split}
g(\nabla_{X'}\underline{L},Y')=& g\Big(\nabla_{\gamma(X){L}+X}\underline{L},\gamma(Y){L}+Y\Big)\\ =&
\gamma(X)\cdot g(\nabla_{L}\underline{L},Y)+\gamma(Y)\cdot g(\nabla_{X}\underline{L},L)+g(\nabla_{X}\underline{L}, Y)\\
=&-\gamma(X)\cdot \zeta(Y)-\gamma(Y)\cdot \zeta(X)+\underline{\chi}(X,Y) 
\end{split}\end{equation*}
and
\begin{equation*}\begin{split}
g(\nabla_{X'}P, Y')=&g\Big(\nabla_{\gamma(X)L+X}P,\gamma(Y)L+Y\Big)
\\=&\gamma(X)\cdot g(\nabla_{L}P,Y)-\gamma(Y)\cdot \chi(P,X)+g(\nabla_{X}P,Y)
\\=&\gamma(X)\cdot (\nabla_{L}\gamma)(Y)-\gamma(Y)\cdot \chi(\gamma^{\sharp},X)+(\nabla_{X}\gamma)(Y),
\end{split}\end{equation*}
and hence
\[\Omega\underline{\chi}'(X',Y')=\a\cdot \chi(X,Y)-\gamma(X)\cdot \zeta(Y)-\gamma(Y)\cdot \zeta(X)+\underline{\chi}(X,Y)+\gamma(X)\cdot (\nabb_{L}\gamma)(Y)-\gamma(Y)\cdot( \chi\cdot\gamma^{\sharp})(X)+(\nabb_{X}\gamma)(Y). 
   \]
   and so
   \[\Omega\underline{\chi}'=\frac{1}{2}(\ga,\ga)\cdot\chi-\ga\otimes\zeta-\zeta\otimes\ga+\underline{\chi}+\ga\otimes\nabb_{L}\ga-\ga\otimes \chi\cdot \ga^{\sharp}+\nabb\ga.\] 
Therefore, since the metrics $\gi,\gi'$ are pointwise the same, i.e~at any point we have $g(X,Y)=g(X',Y')$, we obtain
\begin{equation}\begin{split}
\Omega tr'\underline{\chi}'=& \frac{1}{2}\big(\gamma,\gamma\big)\cdot tr\chi-2(\gamma,\zeta)+tr\underline{\chi}+\Big(\nabb_{L}\gamma,\gamma\Big)-\Big(\chi\cdot P, \gamma\Big)+\divv \gamma\\
=& \frac{1}{2}\big(\gamma,\gamma\big)\cdot tr\chi+tr \underline{\chi}+\divv \gamma+\Big(\nabb_{L}\gamma-\chi\cdot \gamma^{\sharp}-2\zeta,\gamma\Big)\\
=&tr\underline{\chi}+\divv\ga-\hat\chi(\ga^{\sharp},\ga^{\sharp})+\Big(\nabb_{L}\ga-2\zeta,\ga\Big)
 \end{split}\label{tra}\end{equation}
We next express the $\mathcal{S}'$-induced operators $\nabb', \ \divv'$ in terms of the $\mathcal{S}$ foliation. Let $p\in S'_{v'}\cap S_{v}$ and, as before,  let $\psi$ be a smooth function on $\hh$. We have $Y'\psi=Y\psi+\gamma(Y)\cdot L\psi$ for all $Y\in TS_{v}$. Assume that
\[\nabb'\psi=X+\gamma(X)\cdot L, \]for some $X\in TS_{v}$.  
Then for all $Y\in TS_{v}$ we have
\[g(X,Y)=g(Y,\nabb'\psi)=g(Y',\nabb'\psi)=Y'\psi=Y\psi+\gamma(Y)\cdot L\psi=g(\nabb \psi+\gamma^{\sharp}\cdot L\psi, Y)\]
and since this is true for all $Y\in  TS_{v}$ and since by definition $X\in TS_{v}$ we have $X=\nabb \psi+\gamma^{\sharp}\cdot L\psi$ and so
\begin{equation}
{\nabb'\psi=\left[\nabb \psi+\gamma^{\sharp}\cdot L\psi\right]+\left[\gamma^{\sharp}\cdot \nabb \psi+L\psi\cdot(\gamma,\ga)\right]\cdot L.
}
\label{nabla}\end{equation}
We next compute $\divv'\xi'$ for a $S'$ vector field $\xi'$. Let $E_{i},i=1,2$ be a local orthonormal $S$ frame. We have
\begin{equation}
\begin{split}
\divv' \xi'=&\sum_{i}g(\nabb'_{E_{i}'}\xi',E_{i}')=\sum_{i}g(\nabla_{E_{i}'}\xi',E_{i}')=\sum_{i}g\left(\nabla_{(E_{i}+\ga_{i}L)}\xi+\ga(\xi)L,E_{i}+\ga_{i}L\right)\\
=&\sum_{i}g\left(\nabla_{E_{i}}\xi,E_{i} \right)+g\left(\nabla_{E_{i}}\xi,\ga_{i} L\right)+\ga_{i}g\left(\nabla_{L}\xi,E_{i}\right)+g\left(\nabla_{E_{i}}\ga(\xi)L,E_{i}\right)\\
=&\sum_{i}g\left(\nabla_{E_{i}}\xi,E_{i} \right)+\ga_{i}g\left(\nabla_{E_{i}}\xi, L\right)+\ga_{i}g\left(\nabla_{L}\xi,E_{i}\right)+\ga(\xi)g\left(\nabla_{E_{i}}L,E_{i}\right)\\
=&\sum_{i}g\left(\nabla_{E_{i}}\xi,E_{i} \right)-\ga_{i}g\left(\xi, \nabla_{E_{i}}L\right)+\ga_{i}g\left(\nabla_{L}\xi,E_{i}\right)+\ga(\xi)g\left(\nabla_{E_{i}}L,E_{i}\right)\\
=&\divv\xi +\sum_{i}-\ga_{i}\chi(E_{i},\xi)+\ga_{i}g\left(\nabla_{L}\xi,E_{i}\right)+\ga(\xi)\chi(E_{i},E_{i})\\
=&\divv\xi-\chi(\xi,\ga^{\sharp})+\ga(\xi)\cdot tr\chi+g(\nabb_{L}\xi,\ga^{\sharp})
\\
=&\divv\xi-\hat{\chi}(\xi,\ga^{\sharp})+\frac{1}{2}\ga(\xi)\cdot tr\chi+g(\nabb_{L}\xi,\ga^{\sharp})\\
=&\divv\xi+\ga(\xi)\cdot tr\chi+\ga\Big(\li_{L}\xi\Big),
\end{split}
\label{divergence}
\end{equation} 
where we used that $g([\xi,E_{i}],L)=0$. We also have
\begin{equation}
\begin{split}
\li_{L}(\nabb \psi )=&\li_{L}\Big(\gi^{\mu\nu}\cdot \partial_{\mu}\psi \cdot \partial_{\nu} \Big)=\Big(\li_{L}\gi^{\mu\nu}\Big)\cdot \partial_{\mu}\psi \cdot \partial_{\nu} +\gi^{\mu\nu}\cdot \Big(\li_{L} \partial_{\mu}\psi \Big)\cdot \partial_{\nu} +\gi^{\mu\nu}\cdot \li_{L} \partial_{\mu}\psi \cdot[L,\partial_{\nu}] \\
=& -2\chi^{\mu\nu}\cdot\partial_{\mu}\psi \cdot\partial_{\nu}+\gi^{\mu\nu}\cdot \partial_{\mu}(L\psi )\cdot\partial_{\nu}\\
=&-2\chi^{\sharp}\cdot\di \psi +\nabb(L\psi ). 
\end{split}
\label{lieLnabla}
\end{equation} 
Note that we used that $\li_{L}\gi^{\mu\nu}=-2\chi^{\mu\nu}$ and that  the foliation $\mathcal{S}$  is geodesic (and hence $\Omega=1$ in this case).
Regarding the $\mathcal{S}'$-Laplacian  $\lapp'\psi $ we have
\begin{equation}
\begin{split}
\lapp'\psi =&\divv'\nabb'\psi = \divv(\nabb \psi +\ga^{\sharp}\cdot L\psi )+\ga(\nabb \psi +\ga^{\sharp}\cdot L\psi )\cdot tr\chi+\ga\Big(\li_{L}(\nabb \psi +\ga^{\sharp}\cdot L\psi )\Big)\\
=&\lapp \psi +2\big(\ga,\di(L\psi )\big)+L\psi \cdot\divv\gamma+tr\chi\cdot(\ga,\di \psi )+tr\chi\cdot L\psi \cdot (\ga,\ga)-2\chi(\ga^{\sharp},\nabb \psi )\\&+LL\psi \cdot(\ga,\ga)+L\psi \cdot\gi\big(\li_{L}\gamma^{\sharp},\ga^{\sharp}\big).
\end{split}
\label{laplacians}
\end{equation} 
Note that, since $tr\chi\cdot(\ga,\di\psi)-2\chi(\ga^{\sharp},\nabb\psi)=2\hat{\chi}(\ga^{\sharp},\nabb\psi)$, if we let $L\psi=0$ then we obtain that we have $\lapp'\psi=\lapp\psi$ if and only if $\hat{\chi}=0$ which, by virtue of \eqref{conformalequation}, is equivalent to the fact that the induced metrics $\gi,\gi'$ are conformal. 

We next compute $A'\psi=\left[\nabb'\Omega^{2}+2\Omega^{2}\cdot(\zeta')^{\sharp}\right]\cdot\nabb'\psi$. Assuming that $(\zeta')^{\sharp}=X+\ga(X)\cdot L$ and $Y\in TS_{v}$ we have
\begin{equation*}
\begin{split}
g(X,Y)=g((\zeta')^{\sharp},Y')=\zeta'(Y')=g\Big(-(L\log\Omega)\cdot\ga^{\sharp}-\nabb\log\Omega+\zeta^{\sharp}+\chi^{\sharp}\cdot\gamma^{\sharp},Y\Big)
\end{split}
\end{equation*}
and so
\begin{equation}
X=-(L\log\Omega)\cdot\ga^{\sharp}-\nabb\log\Omega+\zeta^{\sharp}+\chi^{\sharp}\cdot\gamma^{\sharp}.
\label{zeta}
\end{equation}
Therefore, recalling the expressions for $\nabb'\Omega^{2},\, \nabb'\psi$ we obtain
\begin{equation*}
\begin{split}
A'\psi=&\Big[\nabb\Omega^{2}+(L\Omega^{2})\cdot\ga^{\sharp}+2\Omega^{2}\cdot\left[-(L\log\Omega)\cdot\ga^{\sharp}-\nabb\log\Omega+\zeta^{\sharp}+\chi^{\sharp}\cdot\gamma^{\sharp}\right] \Big]\cdot\big[\nabb\psi+(L\psi)\cdot \ga^{\sharp}\big]\\
=&2\Omega^{2}\cdot\left[\zeta^{\sharp}+\chi^{\sharp}\cdot\gamma^{\sharp}\right]\cdot\big[\nabb\psi+(L\psi)\cdot \ga^{\sharp}\big]\\
=&2\Omega^{2}\cdot\left[\zeta^{\sharp}\cdot\nabb\psi+(L\psi)\cdot(\ga,\zeta)+\chi(\ga^{\sharp},\nabb\psi)+(L\psi)\cdot \chi(\ga^{\sharp},\ga^{\sharp})\right]. 
\end{split}
\end{equation*}
We thus have
\begin{equation}
\lapp'\psi+\left[\nabb'\Omega^{2}+2\Omega^{2}\cdot(\zeta')^{\sharp}\right]\cdot\nabb'\psi=\lapp\psi+2\zeta^{\sharp}\cdot\nabb\psi+\mathcal{E}\psi,
\label{higherorder}
\end{equation}
where
\begin{equation}
\begin{split}
\mathcal{E}\psi=&2\big(\ga,\di(L\psi )\big)+L\psi \cdot\divv\gamma+tr\chi\cdot(\ga,\di \psi )+tr\chi\cdot L\psi \cdot (\ga,\ga)\\&+LL\psi \cdot(\ga,\ga)+L\psi \cdot\gi\big(\li_{L}\gamma^{\sharp},\ga^{\sharp}\big)+2(L\psi)\cdot(\ga,\zeta)+2(L\psi)\cdot \chi(\ga^{\sharp},\ga^{\sharp}).
\end{split}
\label{epsilon}
\end{equation}

We next look at the coefficient of the zeroth order term in the expression for $\o'_{v'}\psi$: 
\begin{equation}
w'=2\divv'\Big(\Omega^{2}\cdot \zeta'\Big)+\partial_{v'}\big(\Omega tr'\underline{\chi}'\big)+\frac{1}{2}(\Omega tr'\chi')(\Omega tr'\underline{\chi}').
\label{sigmatonos}
\end{equation}
In view of \eqref{chichi} we have
\begin{equation}
\Omega tr'\chi'=\Omega^{2}tr\chi.
\label{trchi}
\end{equation}
and \eqref{tra} can also be written as
\begin{equation}
\begin{split}
\Omega tr'\underline{\chi}'=& tr\underline{\chi}+\divv\ga +\gi(\nabb_{L}\ga^{\sharp},\ga^{\sharp})-2(\zeta,\ga)-\hat{\chi}(\ga^{\sharp},\ga^{\sharp})\\
=& tr\underline{\chi}+\divv\ga +\ga\big(\li_{L}\ga^{\sharp}\big)+\frac{1}{2}tr\chi\cdot(\ga,\ga)-2(\zeta,\ga).\\
\end{split}
\label{trachibar}
\end{equation}
We next compute $\divv'\big(2\Omega^{2}\zeta'\big)$. In view of \eqref{zeta} we obtain that if 
\[ 2\Omega^{2}\zeta'= X_{\Omega}+\ga(X_{\Omega})\cdot L, \]
then
\[X_{\Omega}=-(L\Omega^{2})\cdot \gamma^{\sharp}-\nabb\Omega+2\Omega^{2}\cdot\zeta^{\sharp}+2\Omega^{2}\chi^{\sharp}\cdot\ga^{\sharp}. \]
Therefore,
\begin{equation*}
\begin{split}
\divv'\big(2\Omega^{2}\zeta'\big)=&\divv X+\ga(X)\cdot tr\chi+\ga([L,X])\\
=&-\di L\Omega^{2}\cdot \ga^{\sharp}-L\Omega^{2}\cdot \divv\ga-\lapp\Omega^{2}+2(\di\Omega^{2},\zeta)+2\Omega^{2}\cdot\divv\zeta+2\divv\Big(\Omega^{2}\cdot\chi^{\sharp}\cdot\ga^{\sharp}\Big)\\
&+\left[-(L\Omega^{2})\cdot(\ga,\ga)-(\ga^{\sharp},\nabb\Omega^{2})+2\Omega^{2}\cdot(\zeta,\gamma)+2\Omega^{2}\cdot\chi(\gamma^{\sharp},\ga^{\sharp})\right]\cdot tr\chi\\&+\ga\Big(\li_{L}X\Big).
\end{split}
\end{equation*}
Note the following
\[\divv\big(\chi(\ga^{\sharp},\cdot)\big)=(\divv\chi)(\ga^{\sharp})+(\chi,\nabb\ga^{\sharp}), \]
\[\li_{L}X=-(LL\Omega^{2})\cdot\ga^{\sharp}-(L\Omega^{2})\cdot \li_{L}\ga^{\sharp}-\li_{L}\nabb\Omega^{2}+2(L\Omega^{2})\cdot\zeta^{\sharp}+2\Omega^{2}\cdot\li_{L}\zeta^{\sharp}+2(L\Omega^{2})\cdot\chi^{\sharp}\cdot\ga^{\sharp}+2\Omega^{2}\cdot\li_{L}\Big(\chi\big(\ga^{\sharp},\cdot\big)^{\sharp}\Big) \]
\[ \li_{L}\big(\nabb \Omega^{2} \big)=\nabb L\Omega^{2}-2\chi^{\sharp}\cdot \di\Omega^{2},\]
\[\ga\Big(\li_{L}\big(\nabb \Omega^{2} \big)\Big)=(\ga,\di L\Omega^{2})-2\chi(\ga^{\sharp},\nabb\Omega^{2}),\]
\[\li_{L}\Big(\chi\big(\ga^{\sharp},\cdot\big)^{\sharp}\Big)=-2\chi^{\sharp}\cdot (\chi\cdot\ga^{\sharp})+\big((\li_{L}\chi)\cdot\ga^{\sharp}\big)^{\sharp}+\chi^{\sharp}\cdot\li_{L}\ga^{\sharp},  \]
\[\ga\bigg(\li_{L}\Big(\chi\big(\ga^{\sharp},\cdot\big)^{\sharp}\Big)\bigg)=-2(\chi\times\chi)(\ga^{\sharp},\ga^{\sharp})+\big(\li_{L}\chi\big)(\ga^{\sharp},\ga^{\sharp})+\chi\big(\li_{L}\ga^{\sharp},\ga^{\sharp}\big). \]
Therefore, we obtain
\begin{equation}
\begin{split}
\divv'\big(2\Omega^{2}\zeta'\big)=&-\di L\Omega^{2}\cdot \ga^{\sharp}-L\Omega^{2}\cdot \divv\ga-\lapp\Omega^{2}+2(\di\Omega^{2},\zeta)+2\Omega^{2}\cdot\divv\zeta+
2\chi(\ga^{\sharp},\nabb\Omega^{2})\\&+2\Omega^{2}\cdot (\divv\chi)(\ga^{\sharp})+2\Omega^{2}\cdot(\chi,\nabb\ga^{\sharp})\\
&+\left[-(L\Omega^{2})\cdot(\ga,\ga)-(\ga^{\sharp},\nabb\Omega^{2})+2\Omega^{2}\cdot(\zeta,\gamma)+2\Omega^{2}\cdot\chi(\gamma^{\sharp},\ga^{\sharp})\right]\cdot tr\chi\\&-(LL\Omega^{2})\cdot(\ga,\ga)-(L\Omega^{2})\cdot \gi\big(\li_{L}\ga^{\sharp},\ga^{\sharp}\big)- (\ga,\di L\Omega^{2})+2\chi(\gamma^{\sharp},\nabb\Omega^{2})+2(L\Omega^{2})\cdot(\zeta,\ga)\\&+2\Omega^{2}\cdot (\li_{L}\zeta^{\sharp},\ga^{\sharp})
+2(L\Omega^{2})\cdot\chi(\ga^{\sharp},\ga^{\sharp})-4\Omega^{2}\cdot(\chi\times\chi)(\ga^{\sharp},\ga^{\sharp})+2\Omega^{2}\cdot\big(\li_{L}\chi\big)(\ga^{\sharp},\ga^{\sharp})\\&+2\Omega^{2}\cdot\chi\big(\li_{L}\ga^{\sharp},\ga^{\sharp}\big).
\end{split}
\label{divergencezeta}
\end{equation}

We next compute $\partial_{v'}(\Omega tr'\underline{\chi}')$. Instead of using equation \eqref{tra} we will use the following geometric equation
\begin{equation}
\nabb_{L'}(\Omega\underline{\chi}')=\Omega \nabb_{4'}(\Omega\underline{\chi}')=\Omega^{2}\nabb'\underline{\eta}'+\Omega^{2}\underline{\eta }'\otimes\underline{\eta}'-\Omega\chi'\times\Omega\underline{\chi}'-R'(\cdot,L',\cdot, \underline{L}') 
\label{miji}
\end{equation}which follows from the relations
\begin{equation*}
\begin{split}
&g(\nabla_{X'}e_{3}',\nabla_{4}Y')=\Big[(\underline{\chi}'\times\chi')+\frac{1}{2}(\underline{\eta}'\otimes\underline{\eta}')-\frac{1}{2}(\eta'\otimes \underline{\eta}')\Big](X',Y'),\\
&g(\nabla_{X'}\nabla_{4}e_{3}',Y')=(\nabla_{X'}\underline{\eta}')(Y')-\omega' \underline{\chi}'(X',Y'),\\
&g\left(\nabla_{[\nabla_{4},X']}e_{3}',Y'\right)=\frac{1}{2}(\eta'\otimes\underline{\eta}')(X',Y')+\frac{1}{2}(\underline{\eta}'\otimes\underline{\eta}')(X',Y'),
\\&g(\nabla_{4}\nabla_{X'}e_{3}',Y')=g(\nabla_{X'}\nabla_{4}e_{3}'+\nabla_{[\nabla_{4},X']}e_{3}',Y')+R'(Y',e_{3}',e_{4}',X'),
\end{split}
\end{equation*}
and
\begin{equation*}
\begin{split}
\left(\nabb_{4}\underline{\chi}'\right)(X',Y')&=\left(\nabla_{4}\underline{\chi}'\right)(X',Y')=\nabla_{4}\left(\underline{\chi}'(X',Y')\right)-\underline{\chi}'(\nabla_{4}X',Y')-\chi(X',\nabla_{4}Y')\\
&=\nabla_{4}\big(g(\nabla_{X'}e_{3}',Y')\big)-\underline{\chi}'(\chi^{'\sharp}\cdot X', Y')-\underline{\chi}'(X', \chi^{'\sharp}\cdot Y')\\
&=g(\nabla_{4}\nabla_{X'}e_{3}',Y')+g(\nabla_{X'}e_{3}',\nabla_{4}Y')-(\underline{\chi}'\times\chi')(Y',X')-(\underline{\chi}'\times \chi')(X',Y')\\
&=(\nabla_{X'}\underline{\eta}')(Y')+(\underline{\eta}'\otimes\underline{\eta}')(X',Y')-(\chi'\times\underline{\chi}')(X',Y')-R'(X',e_{4}',Y',e_{3}')-\omega'\underline{\chi}'(X',Y).
\end{split}
\end{equation*}
Recall now that (see Section \ref{sec:TheGeometryOfNullHypersurfaces}) \[\underline{\eta}'=-\zeta'+\di'\log\Omega.\]
Therefore,
\begin{equation}
\Omega^{2}\underline{\eta }'\otimes\underline{\eta}'=\Omega^{2}\zeta'\otimes\zeta'-\frac{1}{2}\zeta'\otimes\di'\Omega^{2}-\frac{1}{2}\di'\Omega^{2}\otimes\zeta'+\di'\Omega\otimes\di'\Omega
\label{1}
\end{equation}
and 
\begin{equation}
\Omega^{2}\nabb'\underline{\eta}'=-\Omega^{2}\nabb'\zeta'+\Omega^{2}\nabb'^{2}\log\Omega=-\nabb'\big(\Omega^{2}\cdot\zeta'\big)+\di'\Omega^{2}\otimes \zeta'+\Omega^{2}\cdot\nabb'^{2}\log\Omega.
\label{2}
\end{equation}
Therefore, by taking the trace of \eqref{miji} with respect to the metric $\gi'$ and using \eqref{1}, \eqref{2} we obtain
\begin{equation}
\begin{split}
\partial_{v'}\big(\Omega tr'\underline{\chi}'\big)=&-\divv'\big(\Omega^{2}\cdot\zeta'\big)+\Omega\lapp'\Omega+(\Omega\zeta',\Omega\zeta')-(\Omega\chi',\Omega\underline{\chi}')-tr'_{X',Y'}R'(X',L',Y',\underline{L}')
\end{split}
\label{mijitracetonos}
\end{equation}
since
\[\Omega^{2}\lapp'\log\Omega=\Omega\lapp'\Omega-(\nabb'\Omega^{2},\nabb'\Omega^{2}).\]
Moreover, note that for the $S$ foliation we have
\begin{equation*}
\begin{split}
\partial_{v}\big( tr\underline{\chi}\big)=-\divv\zeta+(\zeta,\zeta)-(\chi,\underline{\chi})-tr_{X,Y}R(X,L,Y,\underline{L}).
\end{split}
\end{equation*}

The equations \eqref{sigmatonos} and \eqref{mijitracetonos} yield
\begin{equation}
w'=\divv'(\Omega^{2}\cdot\zeta')+(\Omega\zeta',\Omega\zeta')-(\Omega\chi',\Omega\underline{\chi}')+\frac{1}{2}(\Omega tr'\chi')(\Omega tr'\underline{\chi}')-tr'_{X',Y'}R'(X',L',Y',\underline{L}') +\Omega\cdot\lapp'\Omega
\label{newsigmatonos}
\end{equation}
and similarly we obtain
\begin{equation}
w=\divv\zeta +  {(\zeta,\zeta)}-(\chi,\underline{\chi})+\frac{1}{2}tr\chi\cdot tr\underline{\chi}-tr_{X,Y}R(X,L,Y,\underline{L}).
\label{newsigma}
\end{equation}

We now express $(\Omega\zeta',\Omega\zeta')$ and the remaining quantities in \eqref{newsigmatonos} in terms of the $\mathcal{S}$ foliation. By virtue of \eqref{zeta} we obtain
\begin{equation}
\begin{split}
(\Omega\zeta',\Omega\zeta')=&(L\Omega)^{2}\cdot (\ga,\ga)+2\cdot L\Omega\cdot(\ga,\di\Omega)-L\Omega^{2}\cdot (\zeta,\ga)-L\Omega^{2}\cdot \chi(\ga^{\sharp},\ga^{\sharp})\\&+(\nabb\Omega,\nabb\Omega)-2\Omega\cdot(\zeta,\di\Omega)-2\Omega\cdot \chi(\ga^{\sharp},\nabb\Omega)+\Omega^{2}\cdot(\zeta,\zeta)\\&+2\Omega^{2}\cdot\chi(\ga^{\sharp},\zeta)+\Omega^{2}\cdot(\chi\times\chi)(\ga^{\sharp},\ga^{\sharp}).
\end{split}
\label{3}
\end{equation}
Furthermore,
\begin{equation}
\begin{split}
(\Omega\chi',\Omega\underline{\chi}')=&\frac{1}{2}\Omega^{2}\cdot(\ga,\ga)\cdot(\chi,\chi)-2\Omega^{2}\cdot\chi(\ga^{\sharp},\zeta^{\sharp})+\Omega^{2}\cdot(\chi,\underline{\chi})+\Omega^{2}\cdot\chi\big(\nabb_{L}\ga^{\sharp},\ga^{\sharp}\big)\\&-\Omega^{2}\cdot(\chi\times\chi)(\ga^{\sharp},\ga^{\sharp})+\Omega^{2}\cdot(\chi,\nabb\ga).
\end{split}
\label{4}
\end{equation}
This can be further decomposed in trace and trace-free parts. Regarding the curvature components $R'$ and $R$, decomposed with respect to null frames associated to the foliations $\mathcal{S}$ and $\mathcal{S}'$, respectively, we have
\begin{equation*}
\begin{split}
R'(X',L',Y',\underline{L}')=& R\Big(X+\ga(X)L,\Omega^{2}L, Y+\ga(Y)L,\frac{1}{2}(\ga,\ga)L+\underline{L}+\ga^{\sharp}\Big)\\=&
\frac{1}{2}(\ga,\ga)\cdot\Omega^{2}\cdot R\big(X,L,Y,L\big)-\Omega^{2}\cdot\ga(Y)\cdot R\big(X,L,\ga^{\sharp},L\big)\\
&-\Omega^{2}\cdot\ga(Y)\cdot R\big(X,L,\underline{L},L\big)+\Omega^{2}\cdot R\big(X,L,Y,\underline{L}\big)+\Omega^{2}\cdot R\big(X,L,Y,\ga^{\sharp}\big).
\end{split}
\end{equation*}
Therefore, since $tr'_{X',Y'}=tr_{X,Y}$ we have
\begin{equation}
\begin{split}
tr'_{X',Y'}R'(X',L',Y',\underline{L}')=&
\frac{1}{2}(\ga,\ga)\cdot\Omega^{2}\cdot tr\a-\Omega^{2}\cdot\a(\ga^{\sharp},\ga^{\sharp})-\Omega^{2}\cdot\beta(\ga^{\sharp})\\
&+\Omega^{2}\cdot tr_{X,Y}R\big(X,L,Y,\underline{L}\big)+\Omega^{2}\cdot tr_{X,Y}R\big(X,L,Y,\ga^{\sharp}\big).
\end{split}
\label{5}
\end{equation}
Therefore, in view of \eqref{newsigmatonos}, \eqref{divergencezeta}, \eqref{3}, \eqref{4}, \eqref{5},  \eqref{trchi}, \eqref{trachibar} we obtain
\begin{equation*}
\begin{split}
w'=&\underbrace{-\frac{1}{2}\di L\Omega^{2}\cdot \ga^{\sharp}}_{1}\  \underbrace{-\frac{1}{2}L\Omega^{2}\cdot \divv\ga}_{2}\ \underbrace{-\frac{1}{2}\lapp\Omega^{2}}_{3}\ \underbrace{+(\di\Omega^{2},\zeta)}_{4}+\Omega^{2}\cdot\divv\zeta\ 
\underbrace{+\chi(\ga^{\sharp},\nabb\Omega^{2})}_{5}\\
&+\Omega^{2}\cdot (\divv\chi)(\ga^{\sharp})\  {\underbrace{+\Omega^{2}\cdot(\chi,\nabb\ga)}_{6}}\\
&+\frac{1}{2}\cdot tr\chi\cdot\left[\, \underbrace{-(L\Omega^{2})\cdot(\ga,\ga)}_{7}\ \underbrace{-(\ga^{\sharp},\nabb\Omega^{2})}_{8}\ \underbrace{+2\Omega^{2}\cdot(\zeta,\gamma)}_{9}+2\Omega^{2}\cdot\chi(\gamma^{\sharp},\ga^{\sharp})\right]\\
&\underbrace{-\frac{1}{2}(LL\Omega^{2})\cdot(\ga,\ga)}_{10}\ \underbrace{-\frac{1}{2}(L\Omega^{2})\cdot \gi\big(\li_{L}\ga^{\sharp},\ga^{\sharp}\big)}_{11}\ \underbrace{- \frac{1}{2}(\ga,\di L\Omega^{2})}_{1}\ \underbrace{+\chi(\gamma^{\sharp},\nabb\Omega^{2})}_{5}\ \underbrace{+(L\Omega^{2})\cdot(\zeta,\ga)}_{12}\\
&+\Omega^{2}\cdot (\li_{L}\zeta^{\sharp},\ga^{\sharp})
\ \underbrace{+(L\Omega^{2})\cdot\chi(\ga^{\sharp},\ga^{\sharp})}_{13}\ \underbrace{-2\Omega^{2}\cdot(\chi\times\chi)(\ga^{\sharp},\ga^{\sharp})}_{14}\ \underbrace{+\Omega^{2}\cdot\big(\li_{L}\chi\big)(\ga^{\sharp},\ga^{\sharp})+\Omega^{2}\cdot\chi\big(\li_{L}\ga^{\sharp},\ga^{\sharp}\big)}_{15}\\
&\underbrace{+(L\Omega)^{2}\cdot (\ga,\ga)}_{10}\ \underbrace{+2\cdot L\Omega\cdot(\ga,\di\Omega)}_{1}\ \underbrace{-L\Omega^{2}\cdot (\zeta,\ga)}_{12}\ \underbrace{-L\Omega^{2}\cdot \chi(\ga^{\sharp},\ga^{\sharp})}_{13}\\
&\underbrace{+(\nabb\Omega,\nabb\Omega)}_{3}\,\underbrace{-2\Omega\cdot(\zeta,\di\Omega)}_{4}\ \underbrace{-2\Omega\cdot \chi(\ga^{\sharp},\nabb\Omega)}_{5}+\Omega^{2}\cdot(\zeta,\zeta)\\
&+2\Omega^{2}\cdot\chi(\ga^{\sharp},\zeta)\ \underbrace{+\Omega^{2}\cdot(\chi\times\chi)(\ga^{\sharp},\ga^{\sharp})}_{14}\\
&-\frac{1}{2}\Omega^{2}\cdot(\ga,\ga)\cdot(\chi,\chi)+2\Omega^{2}\cdot\chi(\ga^{\sharp},\zeta^{\sharp})-\Omega^{2}\cdot(\chi,\underline{\chi})\ \underbrace{-\Omega^{2}\cdot\chi\big(\nabb_{L}\ga^{\sharp},\ga^{\sharp}\big)}_{15}\\
&\underbrace{+\Omega^{2}\cdot(\chi\times\chi)(\ga^{\sharp},\ga^{\sharp})}_{14}\ \underbrace{-\Omega^{2}\cdot(\chi,\nabb\ga)}_{6}\\
&+\frac{1}{2}\cdot\Omega^{2}\cdot tr\chi\cdot\left[ tr\underline{\chi}+\divv\ga +\gi\big(\li_{L}\ga^{\sharp},\ga^{\sharp}\big)+\frac{1}{2}tr\chi\cdot(\ga,\ga)\ \underbrace{-2(\zeta,\ga)}_{9}\right]\\
&-\frac{1}{2}(\ga,\ga)\cdot\Omega^{2}\cdot tr\a\ \underbrace{+\Omega^{2}\cdot\a(\ga^{\sharp},\ga^{\sharp})}_{15}+\Omega^{2}\cdot\beta(\ga^{\sharp})\\
&-\Omega^{2}\cdot tr_{X,Y}R\big(X,L,Y,\underline{L}\big)-\Omega^{2}\cdot tr_{X,Y}R\big(X,L,Y,\ga^{\sharp}\big)\\
&\underbrace{+\Omega\cdot\lapp\Omega}_{3}\,  \underbrace{\!+2\Omega\cdot(\ga,\di(L\Omega))}_{1}\ \underbrace{+\frac{1}{2}L\Omega^{2}\cdot\divv\ga}_{2}\ \underbrace{+\frac{1}{2}tr\chi\cdot (\ga,\di\Omega^{2})}_{8}\\
&\underbrace{+\frac{1}{2}tr\chi\cdot (L\Omega^{2})\cdot(\ga,\ga)}_{7}\ \underbrace{-\chi(\ga^{\sharp},\nabb\Omega^{2})}_{5}\ \underbrace{+(LL\Omega)\cdot\Omega\cdot (\ga,\ga)}_{10} \ \underbrace{+{{\frac{1}{2}(L\Omega^{2})\cdot \gi\big(\li_{L}\ga^{\sharp},\ga^{\sharp}\big)}}}_{11}.
\end{split}
\end{equation*}
The terms associated with the same number cancel out. For the curvature component $\a$ we used that \[\Omega^{2}\cdot\big(\li_{L}\chi\big)(\ga^{\sharp},\ga^{\sharp})+\Omega^{2}\cdot \chi\big(\li_{L}\ga^{\sharp},\ga^{\sharp} \big)=\Omega^{2}\cdot\chi(\nabb_{L}\ga^{\sharp},\ga^{\sharp})-\Omega^{2}\cdot\a(\ga^{\sharp},\ga^{\sharp}), \] 
which follows from the second variational formula:
\[\li_{L}\chi=-\a+\chi\times\chi. \]
Note that by the null Codazzi equation corresponding the embedding of $S_{v}$ in $\m$ we obtain
\[tr_{X,Y}R\big(X,L,Y,\ga^{\sharp}\big)=\chi(\ga^{\sharp},\zeta^{\sharp})+(\divv\chi)(\ga^{\sharp})-(tr\chi)\cdot(\zeta,\ga^{\sharp})-(\di tr\chi,\ga^{\sharp}) \] 
and hence 
\begin{equation}
\begin{split}
w'=& \boxed{\Omega^{2}\cdot\divv\zeta}\  \underbrace{+\Omega^{2}\cdot (\divv\chi)(\ga^{\sharp})}_{1}
+\Omega^{2}\cdot tr\chi\cdot\chi(\gamma^{\sharp},\ga^{\sharp})\ \underbrace{+\Omega^{2}\cdot (\li_{L}\zeta^{\sharp},\ga^{\sharp})}_{2} +\boxed{\Omega^{2}\cdot(\zeta,\zeta)}\\
&\underbrace{+2\Omega^{2}\cdot\chi(\ga^{\sharp},\zeta)}_{2}-\frac{1}{2}\Omega^{2}\cdot(\ga,\ga)\cdot(\chi,\chi)\ \underbrace{+2\Omega^{2}\cdot\chi(\ga^{\sharp},\zeta^{\sharp})}_{2}\boxed{ -\Omega^{2}\cdot(\chi,\underline{\chi})}\\
&+\frac{1}{2}\cdot\Omega^{2}\cdot tr\chi\cdot\left[\boxed{  tr\underline{\chi}}+\divv\ga +\gi\big(\li_{L}\ga^{\sharp},\ga^{\sharp}\big)+\frac{1}{2}tr\chi\cdot(\ga,\ga)\right]-\frac{1}{2}(\ga,\ga)\cdot\Omega^{2}\cdot tr\a\ \ \underbrace{+\Omega^{2}\cdot\beta(\ga^{\sharp})}_{2}\\
&\boxed{ -\Omega^{2}\cdot tr_{X,Y}R\big(X,L,Y,\underline{L}\big)} -\Omega^{2}\cdot \left[\, \underbrace{\chi(\ga^{\sharp},\zeta^{\sharp})}_{2}\ \underbrace{+(\divv\chi)(\ga^{\sharp})}_{1}-(tr\chi)\cdot(\zeta,\ga^{\sharp})-(\di tr\chi,\ga^{\sharp})\right].\\
\end{split}
\label{newestsigma}
\end{equation}The boxed terms add up to $\Omega^{2}\cdot w$. 
For the curvature component $\beta$ we have used that 
\[ \Omega^{2}\cdot \gi\big(\li_{L}\zeta^{\sharp},\ga^{\sharp}\big)=-3\Omega^{2}\cdot \chi(\ga^{\sharp},\zeta^{\sharp})-\Omega^{2}\cdot\beta(\ga^{\sharp}), \]
which follows from
\begin{equation}
\begin{split}
(\nabb_{4}\zeta)(X)&=(\nabla_{4}\eta)(X)=\nabla_{4}(\eta(X))-\eta(\nabla_{4}X)= \nabla_{4}\big(g(\nabla_{3}e_{4},X)\big)-g(\nabla_{3}e_{4},\nabla_{4}X)\\
&=g(\nabla_{4}\nabla_{3}e_{4},X)+g(\nabla_{3}e_{4},\nabla_{4}X)-g(\nabla_{3}e_{4},\nabla_{4}X)\\&=g\big(R(e_{4},e_{3})e_{4}+\nabla_{3}\nabla_{4}e_{4}+\nabla_{[e_{4},e_{3}]}e_{4},X\big)\\
&=R(X,e_{4},e_{4},e_{3})+\omega g(\nabla_{3}e_{4},X)+g\left(\nabla_{\left(\underline{\eta}^{\sharp}-\eta^{\sharp}-\omega e_{3}+\underline{\omega}e_{4}\right)}e_{4},X\right)\\
&=-\beta(X)+\omega \eta(X)+g\left(\nabla_{\left(\underline{\eta}^{\sharp}-\eta^{\sharp}\right)}e_{4},X\right)-\omega\eta(X)\\
&=\left(-\beta+\chi\cdot\left(\underline{\eta}^{\sharp}-\eta^{\sharp}\right) \right)(X)=\Big(-\beta-2\chi\cdot \zeta^{\sharp}\Big)(X),
\end{split}
\label{propagationzeta}
\end{equation}
where we have used that $\eta=-\underline{\eta}=\zeta$ for $\Omega=1$. Equivalently, we obtain
\begin{equation}
\li_{L}\zeta=-\beta -\chi\cdot\zeta^{\sharp}.
\label{eq:zeta}
\end{equation} 
Therefore, \eqref{higherorder},\eqref{epsilon} and \eqref{newestsigma} imply 
\begin{equation}
\begin{split}\o^{\s'}\psi=&\Omega^{2}\cdot\o^{\s}\psi +\mathcal{E}\psi+\mathcal{R}\psi,
\end{split}
\label{10}
\end{equation}
where $\mathcal{R}\psi=R\cdot\psi$ and
\begin{equation*}
\begin{split}
R=& \Omega^{2}\cdot tr\chi\cdot\chi(\gamma^{\sharp},\ga^{\sharp})-\frac{1}{2}\Omega^{2}\cdot(\ga,\ga)\cdot(\chi,\chi)
+\frac{1}{2}\cdot\Omega^{2}\cdot tr\chi\cdot\left[\divv\ga +\gi\big(\li_{L}\ga^{\sharp},\ga^{\sharp}\big)+\frac{1}{2}tr\chi\cdot(\ga,\ga)\right]\\&-\frac{1}{2}(\ga,\ga)\cdot\Omega^{2}\cdot tr\a\ 
 -\Omega^{2}\cdot \left[-(tr\chi)\cdot(\zeta,\ga^{\sharp})-(\di tr\chi,\ga^{\sharp})\right].\\
\end{split}
\end{equation*}
A straightforward computation gives
\begin{equation}
\mathcal{E}\psi +\mathcal{R}\psi =\Omega^{2}\cdot \mathcal{B}\psi \cdot\left[2\chi(\ga^{\sharp},\ga^{\sharp})+2(\zeta,\ga)+\gi(\li_{L}\ga^{\sharp},\ga^{\sharp})+\divv\ga\right]+\Omega^{2}\cdot(\ga,\ga)\cdot \mathcal{B}\big(\mathcal{B}\psi \big)+ \big(2\gamma,\di(\mathcal{B}\psi )\big),
\label{11}
\end{equation}
where \[\mathcal{B}\psi =L\psi +\frac{1}{2}tr\chi\cdot\psi .\]
If we now set $\psi=\frac{1}{\phi}\cdot\Psi$ with $L\Psi=0$, then 
\[L\psi=-L\phi\cdot \frac{1}{\phi^{2}}\cdot\Psi=-\frac{1}{2}tr\chi\cdot\psi\]
and so 
\begin{equation}
\mathcal{B}\left(\frac{1}{\phi}\cdot\Psi\right)=0.
\label{12}
\end{equation}
Clearly, \eqref{10}, \eqref{11} and \eqref{12} imply \eqref{toshow}, which, using \eqref{oa}, implies the desired \eqref{fora}, which completes the proof for the case where $f=1$. Let us now return to the general case. 

  Let, therefore,  $\mathcal{S}_{1},\mathcal{S}_{2}$ be the foliations given by \eqref{s11}, \eqref{s21}. Consider the following auxiliary foliation
\[\mathcal{S}_{aux}=\Big\langle S_{1}, \big(L_{geod}\big)_{2}, \Omega_{aux}\Big\rangle,  \]
where 
\[\Omega_{aux} =  \frac{1}{f}\cdot\Omega_{{1}}, \]
and $f$ is given by \eqref{f}. 
Let $\mathcal{A}^{aux},\o^{aux}$ be the operators associated to $\mathcal{S}_{aux}$. We will show that 
\begin{equation}
\o^{aux}\big(\psi\big)=\o^{\mathcal{S}_{1}}\left(\frac{1}{f^{2}}\cdot\psi\right)
\label{newwanted}
\end{equation}
for $\psi\in C^{\infty}(\hh)$. Observe that although the geodesic vector fields and the null lapse functions do not agree for the foliations $\mathcal{S}_{aux}$ and $\mathcal{S}_{1}$, the sections of these foliations are the same. Indeed, we have 
\[L_{aux}=\Omega^{2}_{aux}\cdot \big(L_{geod}\big)_{2}=(\Omega_{1})^{2}\cdot \big(L_{geod}\big)_{1}= L_{\mathcal{S}_{1}}.\]
Therefore, equation \eqref{newwanted} can be obtained by a straightforward calculation given that 
\[ (e_{4})_{aux}=f\cdot (e_{4})_{\mathcal{S}_{1}},\  (e_{3})_{aux}=f\cdot (e_{3})_{\mathcal{S}_{1}},\ \ tr\chi_{aux}=f\cdot tr\chi_{\mathcal{S}_{1}}, \   tr\underline{\chi}_{aux}=f\cdot tr\underline{\chi}_{\mathcal{S}_{1}}, \  \zeta_{aux}=\zeta-\di\log f.   \]
 Let now $\Psi\in\vh$.   Using \eqref{oa} and \eqref{newwanted} we obtain
\begin{equation*}
\begin{split}
\mathcal{A}^{\mathcal{S}_{1}}\left(\Psi\right) = & \left(\frac{\Omega_{aux}}{\Omega_{{1}}}\right)^{2}\cdot \mathcal{A}^{aux}(f^{2}\cdot\Psi)\overset{\eqref{fora}}{=} \left(\frac{\Omega_{aux}}{\Omega_{{1}}}\right)^{2}\cdot  \mathcal{A}^{\mathcal{S}_{2}}\big(f^{2}\cdot\Psi\big)\\=&
\frac{1}{f^{2}}\cdot \mathcal{A}^{\mathcal{S}_{2}}\big(f^{2}\cdot\Psi\big),
\end{split}
\end{equation*}
where we used that $f^{2}\cdot\Psi$ is constant along the null generators of $\hh$. 

\end{proof}

 We next present some corollaries in terms of the operator $\o^{\s}$ (see \eqref{o}):
\begin{corollary}
Let $\mathcal{S},\mathcal{S}'$ be two foliations of $\hh$ as defined in \eqref{s1}, \eqref{s2}. Then for all functions $\Psi\in\vh$ we have 
\[\o^{\s'}\left(\frac{1}{\phi}\cdot\Psi\right)=\Omega^{2}\cdot\o^{\s}\left(\frac{1}{\phi}\cdot \Psi\right)  \]
on $\hh$.
\label{str}
\end{corollary}

\begin{corollary}
Let $\hh$ be a null hypersurface with vanishing  expansion, i.e.~$tr\chi=0$. Let also $\mathcal{S}$ and $\mathcal{S}'$ be two foliations of $\hh$. Then for all $\Psi\in\vh$ we have
\[\o^{\s'}\big(\Psi\big)=\Omega^{2}\cdot\o^{\s}\big(\Psi\big) \]
on $\hh$. 
\label{noexpansion}
\end{corollary}

\begin{corollary}
Let $S$ and $S'$ be two foliations of $\hh$ with a common section $S_{v_{0}}$. Then
\[\o_{v_{0}}^{\s'}\big(\Psi\big)=\Omega^{2}\cdot\o_{v_{0}}^{\s}\big(\Psi\big) \]
for all functions $\Psi\in C^{\infty}(S_{v_{0}})$. 
\label{cormore}
\end{corollary}

All the previous results hold pointwise.  The following corollary  yields a result about the Kernel of the elliptic operator of $\o^{\s'}_{v'}$ in terms of the Kernels of $\o^{\s}_{v}$:

\begin{corollary}
Let $\s_{1},\s_{2}$ be the foliations given by \eqref{s11},\eqref{s21}, respectively, and the function $f$ given by \eqref{f}.
If $\Psi\in\vh$ then we have
\[\Psi\in Ker\big(\mathcal{A}_{v}^{\s_{1}}\big)\text{ for all }v\in\mathbb{R} \ \ \ \Rightarrow \ \ \ f^{2}\cdot\Psi\in Ker\big(\mathcal{A}_{v'}^{\s_{2}}\big)\text{ for all }v'\in\mathbb{R}\]
and 
\[\frac{1}{\phi}\cdot{\Psi}\in Ker\big(\mathcal{O}_{v}^{\s_{1}}) \text{ for all }v\in\mathbb{R}\ \ \ \Rightarrow \ \ \ f^{2}\cdot\frac{1}{\phi}\cdot{\Psi}\in Ker(\mathcal{O}_{v'}^{\s_{2}})\text{ for all }v'\in\mathbb{R}.\]
\label{corkernel}
\end{corollary}
\begin{proof}
Let $\Psi\in\vh$ such that 
\[\mathcal{A}_{v}^{\s_{1}}\Psi= \o_{v}^{\s_{1}}\left(\frac{1}{\phi}\cdot \Psi\right)=0\]
for all $v\in\mathbb{R}$. Then, 
\begin{equation}
\mathcal{A}_{v'}^{\s_{2}}\,\big(f^{2}\cdot\Psi\big)=\o_{v'}^{\s_{2}}\left(\frac{1}{\phi}\cdot f^{2}\cdot\Psi\right)=0
\label{indeed}
\end{equation}
for all $v'\in\mathbb{R}$. The corollary  can be shown by  sweeping any section $S'_{v'}$ with leaves of the $\mathcal{S}$ foliation as depicted below:
 \begin{figure}[H]
   \centering
		\includegraphics[scale=0.127]{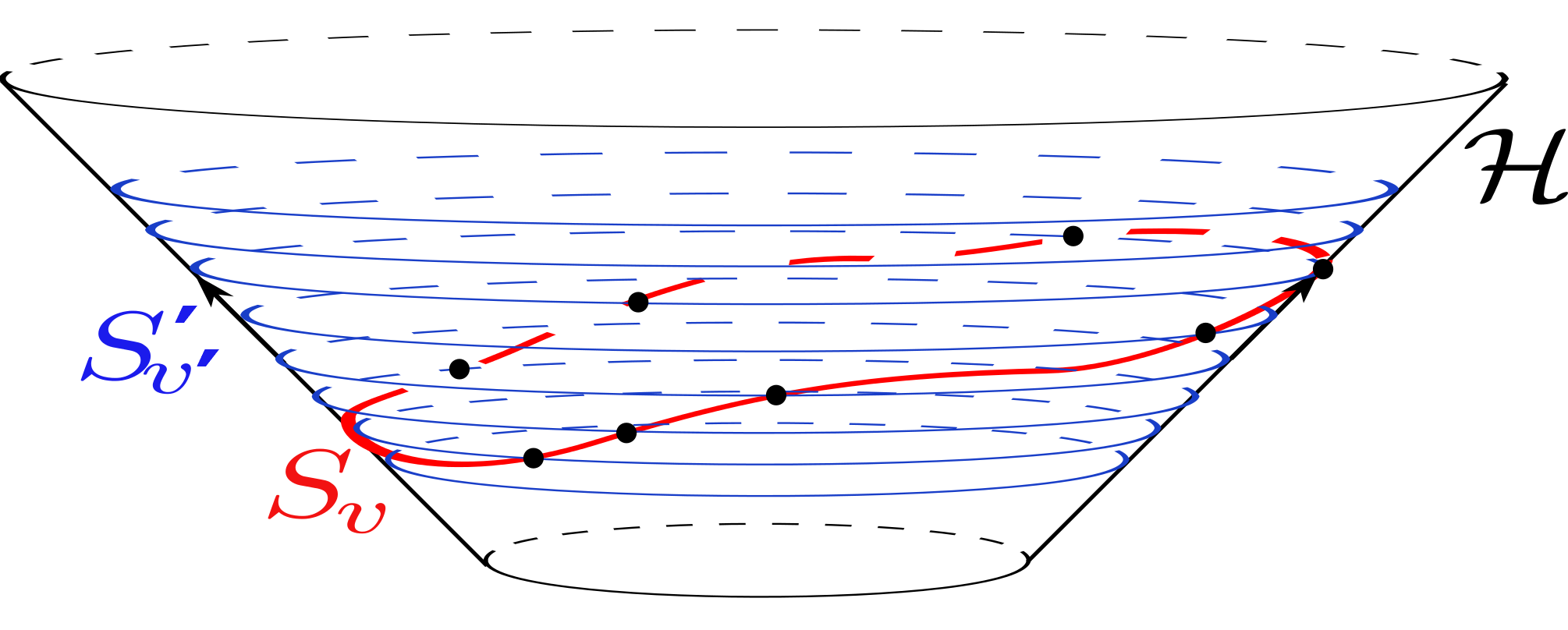}
	\label{fig:nullsection2}
\end{figure}
and applying Theorem \ref{maintheorem} for all points of the section $S'_{v'}$.
\end{proof}

\section{Killing horizons and black holes}
\label{sec:KillingHorizons}

In this section we investigate the operators $\mathcal{A}^{\s},\o^{\s}$ in the context of black hole spacetimes. Specifically, we consider the case where $\hh$ is the event horizon of a black hole. We assume that $\hh$ is a Killing horizon, that is there is a Killing vector field $\xi$ normal to $\hh$. Note that in this case $\xi$ satisfies
\begin{equation}
\nabla_{\xi}\xi=\kappa\cdot \xi, \text{ on  }\hh,
\label{surface}
\end{equation}
where $\kappa$ is constant along the null generators of $\hh$. In consistency with the zeroth law of black hole mechanics, we will assume that $\kappa$ is globally constant on $\hh$. Then, the constant $\kappa$ is called the \textit{surface gravity} of $\hh$. If $\kappa=0$, then the null hypersurface $\hh$ is called an \textit{extremal horizon}.

We first show the following
\begin{lemma}
Let $\hh$ be a Killing horizon and $\mathcal{S}=\left\langle S_{0},L_{geod}, \Omega=1\right\rangle$ be a geodesic foliation, as defined in Section \ref{sec:TheGeometryOfNullHypersurfaces}. Assume that $\xi$ is a Killing vector field normal to $\hh$ and such that \eqref{surface} is satisfied. Then, the following  relations hold on $\hh$: 
\begin{enumerate}
	\item $\chi=0$, 
\item $\li_{L}\gi=0$,
	\item $\di\kappa=g(\xi, \underline{L})\cdot \beta$, where the curvature component $\beta$ is given by \eqref{curvcompdeflist},
	\item $\li_{L}\zeta=\nabb_{L}\zeta=-\beta$,
	\item If, in addition, we take $\left.L_{geod}\right|_{S_{0}}=\xi$ and $\kappa$ is constant on $\hh$, then
	
	$\li_{L}\underline{\chi}=\nabb_{L}\underline{\chi}=\displaystyle\frac{\kappa}{f}\cdot\underline{\chi}$, where $f$ is such that $\xi=f\cdot {L}$ on $\hh$.

\end{enumerate}
\label{lemma}
\end{lemma}
\begin{proof}

\medskip

\noindent 1. Since $\Omega=1$ we have $L=L_{geod}$, where the vector field $L$ is defined by \eqref{eq:l}. We define the following $(0,2)$ tensor field on $\hh$:
\[\chi_{\xi}(X,Y)=g(\nabla_{X}\xi,Y),\]
where $X,Y\in T_{p}\hh,\, p\in \hh$. Note that since $X,Y$ are tangential to $\hh$, the values of $\chi_{\xi}$ depend only on the restriction of $\xi$ on $\hh$. Therefore, since $\xi$ is Killing, the tensor $\chi_{\xi}$ is  antisymmetric. Moreover, $\xi$ and $L$ are both normal to $\hh$ and hence $\xi=f\cdot L$ where  $f=-g(\xi,\underline{L})$ on $\hh$. Then
\begin{equation}
\chi_{\xi}(X,Y)=g\big(\nabla_{X}(fL),Y\big)=f\cdot\chi(X,Y).
\label{chixi}
\end{equation}
Therefore, in view of the symmetry of $\chi$, the tensor field $\chi_{\xi}$ is also a symmetric (0,2) tensor field on $\hh$. Hence, $\chi=0$.

\medskip

\noindent 2. It is an immediate corollary of the first variational formula.

\medskip

\noindent 3. 
Recall that $\kappa$ is constant along the null generators of $\hh$.  Hence, $L\kappa=0$. 
 
Since $\xi$ is Killing it satisfies
\begin{equation}
\nabla^{2}_{X,Y}\xi:=\nabla_{X}\nabla_{Y}\xi-\nabla_{\nabla_{X}{Y}}\xi=R(X,\xi)Y.
\label{eq:kill}
\end{equation}
If $X,Y$ are tangential to $C$, then we obtain
\begin{equation*}
\begin{split}
g(\nabla_{X}Y,L)=-g(Y,\nabla_{X}L)=0
\end{split}
\end{equation*}
since $\chi=0$. Hence $\nabla_{X}Y$ is also tangential to $\hh$ and thus all the terms in $\eqref{eq:kill}$ depend only on the restriction of $\xi$ on $\hh$. Let us now assume that $X$ is tangential to the sections $S_{\tau}$ of the affine foliation of  $\hh$ and $Y=\xi$. Then \eqref{eq:kill} becomes
\[\nabla_{X}\nabla_{\xi}\xi-\nabla_{\nabla_{X}\xi}\xi=R(X,\xi)\xi.  \]
Note that $g(\nabla_{X}\xi,\xi)=\frac{1}{2}X\big(g(\xi,\xi)\big)=0$ and that, if $Z$ is any $S$ vector field then $g(\nabla_{X}\xi,Z)=\chi_{\xi}(X,Z)=0$. Therefore, 
\[\nabla_{X}\xi=\mu({X})\cdot\xi\] 
for some 1-form $\mu$ on $\hh$, which depends on  the function $f$ and the torsion $\zeta$. Then, \eqref{eq:kill} becomes
\[R(X,\xi)\xi=\nabla_{X}(\kappa\cdot \xi)-\mu(X)\nabla_{\xi}\xi=\big(\nabla_{X}\kappa\big)\cdot \xi+\kappa \nabla_{X}\xi-\mu(X)\nabla_{\xi}\xi=\big(\nabla_{X}\kappa\big)\cdot\xi.  \]
Taking the inner product with $\underline{L}$ we obtain
\[f^{2}\cdot g\big(R(X,L)L,\underline{L}\big)=-f\nabla_{X}\kappa  \]
and hence,
\[\nabla_{X}\kappa=-f\cdot R(\underline{L},L,X,L)=g(\xi,\underline{L})\cdot R(X,L,\underline{L},L)=g(\xi,\underline{L})\cdot \beta(X).\]

\medskip

\noindent 4. Immediate from equation \eqref{eq:zeta} and $\chi=0$.

\medskip

\noindent 5. Under the additional assumptions we have that $\left.f\right|_{S_{0}}=1$ and $Lf=\kappa$. Therefore, since by assumption  $\di\kappa=0$, we have $\di f=0$. Let $X,Y$ be $S$ tangential normal Jacobi vector fields, i.e.~$[L,X]=[L,Y]=0$. Then
\begin{equation}
[\xi,X]=-(Xf)\cdot L=0, \ \ \ [\xi,Y]=-(Yf)\cdot L=0.
\label{fornow}
\end{equation}
Observe now that 
\begin{equation*}
\begin{split}
\nabla_{X}Y=&\nabb_{X}Y+\underline{\chi}(X,Y)\cdot L+\chi(X,Y)\cdot \underline{L}\\ &\nabb_{X}Y+\underline{\chi}(X,Y)\cdot L\in T\hh
\end{split}
\end{equation*}
and that $\lie_{L}\big(\nabb_{X}Y\big)=0$ since $\li_{L}\gi=2\chi=0$. Therefore, 
\begin{equation}
\lie_{L}\big(\nabla_{X}Y\big)=L\big(\underline{\chi}(X,Y)\big)\cdot L. 
\label{lhs}
\end{equation}
Since $\nabla_{X}Y\in T\hh$, then LHS of \eqref{lhs} depends only on the restriction of $L$ on $\hh$. Therefore, since $L=\frac{1}{f}\cdot \xi$ on $\hh$, we obtain
\begin{equation*}
\begin{split}
\lie_{L}\big(\nabla_{X}Y\big)=&\frac{1}{f}\cdot \lie_{\xi}\big(\nabla_{X}Y\big)+\frac{(\nabla_{X}Y)(f)}{f^{2}}\cdot\xi\\
=&\frac{1}{f}\cdot \nabla_{[\xi,X]}Y+\frac{1}{f}\cdot \nabla_{X}[\xi,Y]+\big(\nabb_{X}Y\big)(f)\cdot \frac{1}{f^{2}}\cdot \xi+\underline{\chi}(X,Y)\frac{(Lf)}{f^{2}}\cdot \xi\\
=& \underline{\chi}(X,Y)\cdot\frac{\kappa}{f}\cdot L, 
\end{split}
\end{equation*}
where we used  in the second line that  $\xi$ is a Killing vector field and in the third line that $\big(\nabb_{X}Y\big)(f)=0$ and \eqref{fornow}. 
\end{proof}
If we trace the last identity of the above lemma we obtain
\[L tr\underline{\chi}=\frac{\kappa}{f}\cdot tr\underline{\chi}.  \]
Since $Lf=\kappa$ we obtain
\begin{equation}
tr\underline{\chi}=\left.tr\underline{\chi}\right|_{S_{0}}\cdot f
\label{trchibarlambda}
\end{equation} and so 
\begin{equation}
L tr\underline{\chi}=\left.tr\underline{\chi}\right|_{S_{0}}\cdot\kappa.
\label{tracechibar}
\end{equation}
 Recalling that $tr\chi=0$ and the definition \ref{o}, we obtain
\begin{equation}
	\begin{split}
\o^{\s}\Psi=\lapp\Psi+2\zeta^{\sharp}\cdot\nabb\Psi+\left[2\divv\,\zeta^{\sharp}+\left.tr\underline{\chi}\right|_{S_{0}}\cdot\kappa\right]\cdot\Psi,
\end{split}
\label{killingoperator}
\end{equation}
with respect to the foliation $\mathcal{S}=\left\langle S_{0}, \left.L_{geod}\right|_{S_{0}}=\left.\xi\right|_{S_{0}},\,  \Omega=1\right\rangle$.  

In view of the zeroth law of black hole mechanics  we may assume that $\kappa$ is constant on $\hh$ and hence, by Lemma \ref{lemma}, we obtain that  $\beta=0$ on $\hh$ and hence $\zeta$ is conserved on $\hi$, i.e.~$\li_{L}\zeta=0$. 	From \eqref{tracechibar}, $Ltr\underline{\chi}$ does not depend on $v$. By virtue of \eqref{lphi}, the conformal factor $\phi$ also does not depend on $v$. Therefore, the operators $\mathcal{A}_{v}^{\s},\mathcal{O}_{v}^{\s}$ do not depend on $v$ (modulo identifying the sections $S_{v}$ with $S_{0}$ via the diffeomorphisms $\Phi_{v}$). 

If we now consider a general foliation $\tilde{\mathcal{S}}=\left\langle S_{0},\left.L_{geod}\right|_{S_{0}}=\left.\xi\right|_{S_{0}},\Omega\right\rangle$, then
\[\frac{1}{\Omega^{2}}\cdot \o^{\tilde{\s}}\Psi=\o^{{\s}}\Psi \]
and hence the operators $\mathcal{A}_{v}^{\tilde{\s}}$, $\frac{1}{\Omega^{2}}\cdot \o^{\tilde{\s}}_{v}\Psi$ do not depend on $v$ (again modulo identifying the sections $S_{v}$ with $S_{0}$ via the diffeomorphisms $\Phi_{v}$). 

\begin{remark}
Given a foliation $\mathcal{S}=\left\langle S_{0}, \left.L_{geod}\right|_{S_{0}}=\left.\xi\right|_{S_{0}},\,  \Omega=1\right\rangle$ we can rewrite the operator $\mathcal{O}^{\s}_{v}$ given by \eqref{killingoperator} as follows
\begin{equation}
\o_{v}^{\s}\Psi=\underbrace{\lapp\Psi+2\zeta^{\sharp}\cdot\nabb\Psi+2\divv\,\zeta^{\sharp}\cdot \Psi}_{\mathcal{K}_{v}^{\s}\Psi}+\underbrace{\left.tr\underline{\chi}\right|_{S_{0}}\cdot\kappa\cdot\Psi}_{\mathcal{T}_{v}^{\s}\Psi},
\label{killingoperator1}
\end{equation}
Note that the section $S_{0}$ can be freely chosen for the foliation $\mathcal{S}$. In view of $\li_{L}\gi=\chi=0$, $L\Psi=0$ and \eqref{zetazetatonos}, \eqref{nabla} and \eqref{laplacians} the operator $\mathcal{K}_{v}^{\s}$ does not depend on the choice of the section $S_{0}$ (again, modulo identifying all sections of $\hh$ via the flow of the null generators). However, in view of \eqref{tra}, the operator $\mathcal{T}_{v}^{\s}$ depends on $S_{0}$. Specifically, if we consider another foliation $\mathcal{S}'=\left\langle S_{0}', \left.L_{geod}\right|_{S_{0}'}=\left.\xi\right|_{S_{0}'},\,  \Omega=1\right\rangle$ then in view of the main Theorem \ref{maintheorem}, and recalling that $L\phi=0$, we have that
\begin{equation}
\o_{v}^{\s}(\Psi)=\frac{1}{f^{2}}\cdot\o_{v}^{\s'}\big(f^{2}\cdot\Psi\big),
\label{changeofsection}
\end{equation}
where $f$ is such that $L_{geod}'=f^{2}\cdot L_{geod}$, where $L_{geod}',L_{geod}$ denote the geodesic vector fields of $\s',\s$, respectively.
\label{remarkgiazeta}
\end{remark}

Recall that for extremal black holes we have $\kappa=0$ and hence 
\begin{equation}
\o_{v}^{\s}\Psi=\lapp\Psi+\divv\big(2\Psi\cdot\zeta\big).
\label{operatorextremal}
\end{equation}
Following the argument of \cite{hj2012} one obtains that $dimKer(\o_{v}^{\s})=1$ for all $v$. Indeed, since $\int_{S_{v}}\o^{\s}_{v}\psi=0$ for all $\psi$, the (unique) positive principal eigenfunction $\Psi$ of $\o_{v}^{\s}$ must lie in the Kernel of $\o_{v}^{\s}$. Since $\o_{v}^{\s}$ does not depend on $v$ we deduce that there is a unique (up to constant factors) smooth function $\Psi\in V_{\hh}$ such that $\o^{\mathcal{S}}(\Psi)=0$.

Summarizing we have shown the following
\begin{proposition}
Let $\hh$ be a Killing horizon with constant surface gravity $\kappa$. Let also $S=\left\langle S_{0},\left.L_{geod}\right|_{S_{0}}=\left.\xi\right|_{S_{0}},\Omega\right\rangle$ be a foliation of $\hh$, as defined in Section \ref{sec:TheGeometryOfNullHypersurfaces}. Then the operators $\mathcal{A}^{\s}_{v},\, \frac{1}{\Omega^{2}}\cdot\o^{\mathcal{S}}_{v}$, given by \eqref{a} and \eqref{o}, respectively, do not depend on $v$ modulo identifying $S_{v}$ with $S_{0}$ via the diffeomorphism $\Phi_{v}$, i.e.
\[\big(\Phi_{v}\big)^{*}\mathcal{A}_{v}^{\s}=\mathcal{A}_{0}^{\s}, \ \ \ \ \big(\Phi_{v}\big)^{*}\left(\frac{1}{\Omega^{2}}\cdot\mathcal{O}_{v}^{\s}\right)=\frac{1}{\Omega^{2}}\cdot\mathcal{O}_{0}^{\s}   .\] Moreover, if $\hh$ is an extremal horizon (i.e.~$\kappa=0$) then $dimKer\big(\mathcal{A}^{\s}_{v}\big)=dimKer\big(\o_{v}^{\mathcal{S}}\big)=1$, for all $v\in\mathbb{R}$. 
\label{lastprop}
\end{proposition}

The above proposition can be used in order to retrieve and in fact generalize the conservation law on the event horizon of extremal black holes for all solutions to the linear wave equation  discovered in  \cite{aretakis4,hj2012,murata2012}. This conservation law coupled with dispersive estimates away from the event horizon forces higher order derivatives of generic solutions to the wave equation to blow up asymptotically along the event horizon (see \cite{aretakis1,aretakis3}).  We remark that the previous result is in stark contrast with the subextremal case for which Dafermos and Rodnianski \cite{enadio,tria} have derived quantitative decay estimates for all higher order derivatives in the exterior region up to and including  the event horizon.

 We also remark that  Theorem \ref{maintheorem} and Corollary \ref{corkernel} imply that the aforementioned conservation law holds with respect to \textbf{all} foliations of the event horizon.  Futher applications are presented in a companion paper \cite{aretakisglue}, where it is shown the above properties play a fundamental role in the analysis of the characteristic initial value problem of  the wave equation on Lorentzian manifolds.

\section{Epilogue: elliptic operators on null hypersurfaces}
\label{sec:EpilogueEllipticOperatorsOnNullHypersurfaces}

In the previous sections we have introduced an elliptic operator on a null hypersurface $\hh$ which is covariant under refoliation of $\hh$. This operator plays an important role in the evolution of the wave equation along $\hh$. Since null hypersurfaces are characteristic surfaces of more general (geometric and tensorial) hyperbolic equations, using the methods of the present paper as well as of \cite{aretakisglue}, one expects to be able to derive similar elliptic operators associated to those equations as well. The study of the present paper suggests that the following general definitions and discussion might be relevant. 

We  start with the following definition

\begin{definition}\textbf{(Geometric elliptic operators on $\hh$):} Let $\s=\big(S_{v}\big)_{v\in\mathbb{R}}$ be a foliation of $\hh$. A second order linear operator $\mathcal{A}^{\s}:C^{\infty}(\hh)\rightarrow\mathbb{R}$ is called a \textit{geometric elliptic operator} if it is tangential to the sections $S_{v}$ of $\s$ such that the restriction  \[\mathcal{A}_{v}^{\s}:=\left.\mathcal{A}^{S}\right|_{S_{v}}:C^{\infty}(S_{v})\rightarrow\mathbb{R}\]  is an elliptic operator on the Riemannian manifold $S_{v}$ which depends only on the geometry of the foliation $\s$, i.e~the first and the second fundamental forms of the sections of $\s$ with respect to the ambient manifold $\m$.  
\label{def1}
\end{definition}
An example of a geometric elliptic operator on $\hh$ is the operator $\lapp^{\mathcal{S}}$ which  at each point $p\in\hh$  is given by the induced Laplacian of the (unique) section $S_{v}$ of the foliation $\mathcal{S}$ that passes through the point $p$. In other words, $\lapp^{\s}$ is defined such that $\left.\lapp^{\mathcal{S}}\right|_{S_{v}}=\lapp_{S_{v}}$, for all $v\in\mathbb{R}$, where $\lapp_{S_{v}}$ denotes the induced Laplacian on $S_{v}$.

\medskip

\textbf{Change of foliation and associated operators:} Let $\s=\big\langle S_{0}, L_{geod}, \Omega  \big\rangle$ be a foliation of $\hh$ and $\mathcal{A}^{\s}$ be a geometric elliptic operator. 
Let $\mathcal{S}'=\big\langle S_{0}', L_{geod}', \Omega'  \big\rangle$ be another foliation of $\hh$. Note that $L_{geod}'=f^{2}\cdot L_{geod}$ for some function $f\in C^{\infty}(\hh)$ which is constant along the null generators of $\hh$. One can consider the associated geometric elliptic operator $\mathcal{A}^{\s'}$ defined in an identical way to $\mathcal{A}^{\s}$ by simply replacing the geometric quantities associated to the foliation $\s$ with those of the foliation $\s'$. For example, if $\mathcal{A}^{\s}=\lapp^{\s}+\zeta^{\sharp}\cdot \nabb^{\s}+ (tr\underline{\chi}\cdot tr\chi)\cdot I$, where $\nabb^{\s}$ is the induced gradient on the sections $S_{v}$ of $\s$, $I$ is the identity operator, $\zeta^{\sharp}$ is the torsion and $tr\chi,tr\underline{\chi}$ are the null mean curvatures of $S_{v}$, then one can define the associated operator $\mathcal{A}^{\s'}=\lapp^{\s'}+\big(\zeta^{\sharp}\big)'\cdot \nabb^{\s'}+ (tr\underline{\chi'}\cdot tr\chi')\cdot I$  such that $\left.\lapp^{\s'}\right|_{S'_{v}}=\lapp_{S'_{v}}$, $\left.\nabb^{\s'}\right|_{S'_{v}}=\nabb_{S'_{v}}$ and $(\zeta^{\sharp})',tr\chi',tr\underline{\chi}'$ are the torsion and the null mean curvatures of $S'_{v'}$, respectively. 

One would ideally want to capture elliptic structures on null hypersurfaces by considering geometric elliptic operators $\mathcal{A}^{\s}$ which, however, do \textbf{not} depend, in an appropriate sense, on the choice of the foliation $\s$. However, we have $\mathcal{A}^{\s}\neq \mathcal{A}^{\s'}$. Indeed, at each point $p\in\hh$ the operator $\mathcal{A}^{\s}$ is tangential to $S_{v}$, where $p\in S_{v}$, whereas the operator $\mathcal{A}^{\s'}$ is tangential to $S'_{v'}$, where $p\in S'_{v'}$. 

Therefore, the only way to ``compare'' the operators $\mathcal{A}^{\s}$ and $\mathcal{A}^{\s'}$ is via the flow of the null generators. Given two section $S_{v},S_{v'}$ passing through the point $p$, one can define the diffeomorphism $\Phi_{v}^{v'}:S_{v}\rightarrow S_{v'}'$ such that if $q\in S_{v}$ then $\Phi_{v}^{v'}(q)$ is the intersection of $S'_{v'}$ and the null generator passing through $q$.
 \begin{figure}[H]
   \centering
		\includegraphics[scale=0.109]{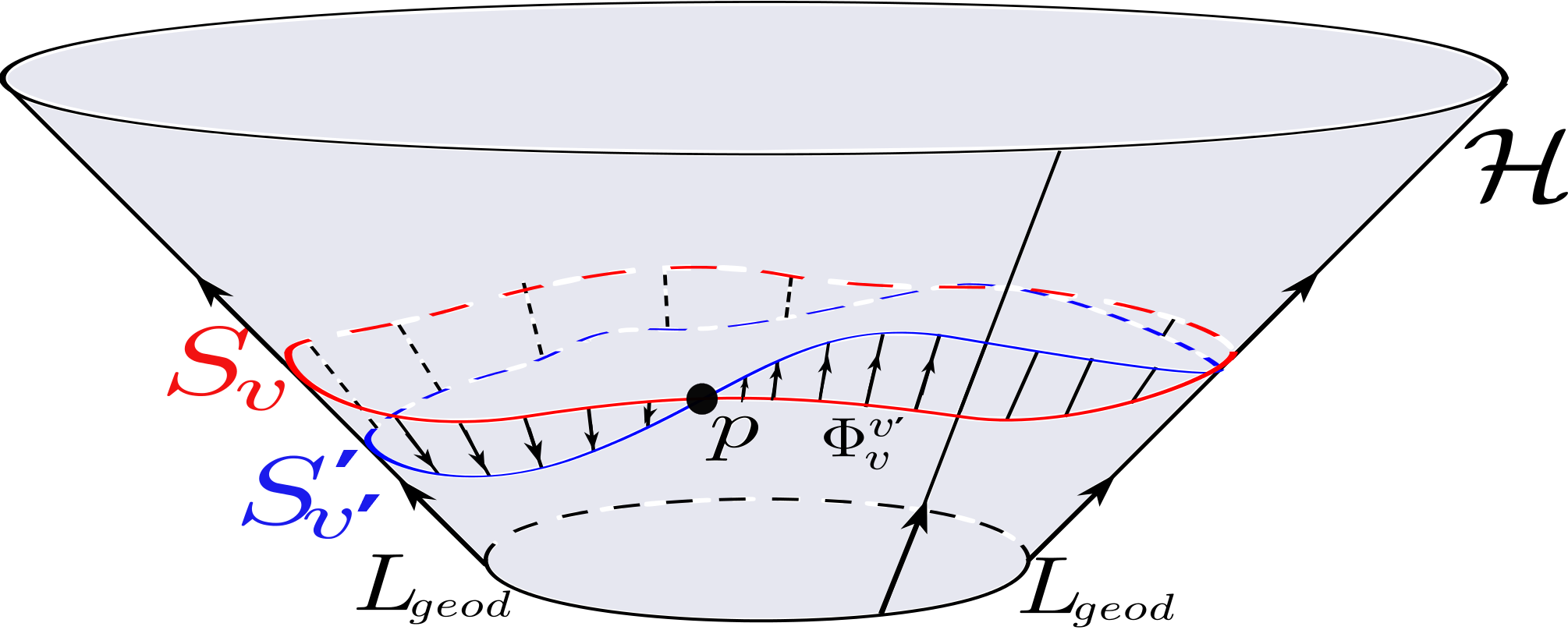}
	\label{fig:nullsection11212jlkj12121}
\end{figure}
The diffeomorphisms $\Phi_{v}^{v'}$ allow us to define the foliation-invariant operators as follows
\begin{definition} \textbf{(Foliation-invariant geometric elliptic operators):} Let $\s=\left\langle S_{0}, L_{geod},\Omega\right\rangle$ be a foliation of $\hh$. A geometric elliptic operator $\mathcal{A}^{S}$ is called foliation-invariant at a point $p\in\hh$ if for any foliation $\s'=\left\langle S'_{0}, L_{geod}'=f^{2}\cdot L_{geod},\Omega'\right\rangle$ of $\hh$ we have  
\[\big(\Phi_{v}^{v'}\big)^{*}\mathcal{A}_{v'}^{\s'}=\mathcal{A}_{v}^{\s} \ \text{  at  }p, \]
where the pullback operator $\big(\Phi_{v}^{v'}\big)^{*}\mathcal{A}_{v'}^{\s'}$ is the operator on $S_{v}$  defined such that $\Big(\big(\Phi_{v}^{v'}\big)^{*}\mathcal{A}_{v'}^{\s'}\Big)(\psi)=\mathcal{A}_{v'}^{\s'}\Big(\big(\Phi_{v}^{v'}\big)_{*}\psi\Big),$
for all $\psi\in C^{\infty}\big(S_{v}\big)$. Furthermore, an operator $\mathcal{A}^{\s}$ is called foliation-invariant if it is foliation-invariant at all points of $\hh$.
\label{definitioninvariantfolellhl}
\end{definition}
In other words,   a foliation-invariant geometric elliptic operator $\mathcal{A}$ satisfies the property that for any two foliations $\s=\left\langle S_{0},L_{geod},\Omega\right\rangle$ and $\s'=\left\langle S'_{0}, L_{geod}, \Omega'\right\rangle$  the operators $\mathcal{A}_{v}^{\s},\mathcal{A}_{v'}^{\s'}$ are exactly the same at $p\in S_{v}\cap S'_{v'}$,  modulo identifying the sections $S_{v},S'_{v'}$ via $\Phi_{v}^{v'}$.

We remark that the Laplacian  $\lapp^{\mathcal{S}}$ is foliation-invariant on a null hypersurface $\hh$ if and only if  the induced metrics of  the sections of $\hh$  are \textit{conformal}, i.e.~the shear of $\hh$ vanishes (see  \eqref{laplacians}). Clearly, the conformal equivalence of the sections of $\hh$ is a very restrictive condition on $\hh$. It turns out that the foliation-invariance is a very strong condition to be satisfied for general null hypersurfaces. For this reason we make the following definition
\begin{definition} \textbf{(Foliation-covariant geometric elliptic operators):} Let $\s=\left\langle S_{0}, L_{geod},\Omega\right\rangle$ be a foliation of $\hh$. A geometric elliptic operator $\mathcal{A}^{S}$ is called foliation-covariant at a point $p\in\hh$ if for any foliation $\s'=\left\langle S'_{0}, L_{geod}'=f^{2}\cdot L_{geod},\Omega'\right\rangle$ of $\hh$ we have  
\[\Big(\big(\Phi_{v}^{v'}\big)^{*}\mathcal{A}_{v'}^{\s'}\Big)=f^{2}\cdot\Big(\mathcal{A}_{v}^{\s}\circ M_{f^{-2}} \Big)\ \text{  at  }p, \]
where the pullback operator $\big(\Phi_{v}^{v'}\big)^{*}\mathcal{A}_{v'}^{\s'}$ is the operator on $S_{v}$  defined such that $\Big(\big(\Phi_{v}^{v'}\big)^{*}\mathcal{A}_{v'}^{\s'}\Big)(\psi)=\mathcal{A}_{v'}^{\s'}\Big(\big(\Phi_{v}^{v'}\big)_{*}\psi\Big),$
for all $\psi\in C^{\infty}\big(S_{v}\big)$, and the operator $M_{f^{2}}$ is defined such that $M_{f^{2}}(\psi)=f^{2}\cdot \psi$ for all $\psi\in C^{\infty}\big(S_{v}\big)$. 
Furthermore, an operator $\mathcal{A}^{\s}$ is called foliation-covariant if it is foliation-covariant at all points of $\hh$.
\label{definitioninvariantfolel}
\end{definition}

It would be very interesting to see if foliation-covariant elliptic operators on $\hh$ associated (in the sense of \cite{aretakisglue}; see also \cite{harveyoperator}) to the Maxwell equations and other linear and non-linear hyperbolic equations  exist. 

\section{Acknowledgements}
\label{sec:Acknowledgements}
	I would like to thank Spyros Alexakis,  Mihalis Dafermos, Harvey Reall and Shiwu Yang for several very helpful discussions and insights. I acknowledge support through NSF grant  DMS-1265538.

\bibliographystyle{acm}
\bibliography{../../../../bibliography}

\end{document}